\documentclass[conference]{IEEEtran}
\IEEEoverridecommandlockouts

\usepackage[T1]{fontenc}
\usepackage[utf8]{inputenc}

\usepackage{amsmath,amssymb}

\usepackage{graphicx}
\usepackage{xcolor}
\usepackage{tikz}
\usetikzlibrary{arrows.meta,positioning}

\usepackage{url}
\usepackage{xspace}
\usepackage{threeparttable}

\usepackage{tcolorbox}

\usepackage{pifont}
\usepackage{wasysym}      

\usepackage{booktabs}
\usepackage{multirow}
\usepackage{tabularx}
\usepackage{array}
\usepackage{colortbl}
\usepackage{makecell}

\usepackage{graphicx}
\usepackage[caption=false,font=footnotesize]{subfig}

\usepackage[section]{placeins}

\makeatletter
\DeclareRobustCommand\onedot{\futurelet\@let@token\@onedot}
\def\@onedot{\ifx\@let@token.\else.\null\fi\xspace}

\def\etal{\emph{et al}\onedot}
\makeatother

\newcommand{\keypoint}[1]{\smallskip\noindent\textbf{#1}\quad}

\newtheorem{definition}{Definition}
\newtheorem{proposition}{Proposition}

\newtheorem{corollary}{Corollary}
\newtheorem{gaps}{Gap}
\newtheorem{insights}{Insight}

\tcbset{
  colback=black!2,
  colframe=black!45
}

\newcommand{\gapbox}[1]{%
  \begin{tcolorbox}[colback=black!2,colframe=black!45]
    \begin{gaps}
      \normalfont #1
    \end{gaps}
  \end{tcolorbox}
}

\newcommand{\insightbox}[1]{%
  \begin{tcolorbox}[colback=blue!2,colframe=blue!35!black]
    \begin{insights}
      \normalfont #1
    \end{insights}
  \end{tcolorbox}
}

\definecolor{coverageFull}{HTML}{4F8582} %
\definecolor{coveragePartial}{HTML}{8A769D} 
\definecolor{coverageNone}{HTML}{D47A20} 
\newcommand{\fullc}{\textcolor{coverageFull}{\ding{108}}}   
\newcommand{\halfc}{\textcolor{coveragePartial}{\LEFTcircle}} 
\newcommand{\emptyc}{\textcolor{coverageNone}{\ding{109}}} 
\newcommand{\outc}{{\textnormal{-}}}

\newcommand{\Cfull}{\fullc}
\newcommand{\Cpart}{\halfc}
\newcommand{\Cnone}{\emptyc}

\usepackage{acro}
\DeclareAcronym{ai}{short = AI, long = Artificial Intelligence}

\DeclareAcronym{llm}{short = LLM, long = Large Language Model}

\newcommand{\LLMs}{\acp{llm}\xspace}
\DeclareAcronym{dae}{short = DAE, long = Decentralized Agentic Economy, long-plural-form = Decentralized Agentic Economies}
\newcommand{\DAE}{\ac{dae}\xspace}
\newcommand{\DAEs}{\acp{dae}\xspace}
\DeclareAcronym{an}{short = AN, long = Agentic Node}

\DeclareAcronym{dao}{short = DAO, long = Decentralized Autonomous Organization}

\DeclareAcronym{gev}{short = GEV, long = Generative Extractable Value}
\newcommand{\GEV}{\ac{gev}\xspace}

\newcolumntype{C}[1]{>{\centering\arraybackslash}m{#1}}
\newcommand{\rot}[1]{\rotatebox[origin=c]{90}{\scriptsize #1}}

\usepackage{hyperref}
\usepackage{cite}
\usepackage{enumitem}

\begin{document}

\bstctlcite{IEEEexample:BSTcontrol}

\author{
\IEEEauthorblockN{
Rui Sun\textsuperscript{1},
Xihan Xiong\textsuperscript{2},
Qin Wang\textsuperscript{3},
Fei Gao\textsuperscript{4},\\
Zelin Li\textsuperscript{5},
Zehua Cheng\textsuperscript{6},
Jiahao Sun\textsuperscript{7},
Zhipeng Wang\textsuperscript{8}}
\IEEEauthorblockA{
\textsuperscript{1}Newcastle University, UK \quad
\textsuperscript{2}University of Bristol, UK \quad
\textsuperscript{3}CSIRO, Australia\\
\textsuperscript{4}University College London, UK \quad
\textsuperscript{5}Ohio State University, USA \quad
\textsuperscript{6}University of Oxford, UK\\
\textsuperscript{7}FLock.io, UK \quad
\textsuperscript{8}The University of Manchester, UK}
}

\title{SoK: Decentralized Agent Economic Infrastructure}

\maketitle

\begin{abstract}

Decentralized agent economies increasingly build a single task from protocols that were designed and secured separately. This creates a simple problem: a workflow can look correct at each step and still produce the wrong outcome. For example, a correct escrow may release payment on an authorized approval that provides little evidence that the delivered work actually satisfied the task.


We systematize this problem across the full lifecycle of an agent task. Our study organizes security and economic requirements into $17$ property families over six stages, with receipt soundness and completeness assessed separately. We examine $12$ systems and standards, five reusable mechanism families, and four classical baselines. We introduce \textit{guarantee closure}, a task-relative criterion for determining whether guarantees established at one stage remain available and constrain the later decisions that depend on them.


We apply the criterion to controlled and native workflows, covering 840 matched executions and an exhaustive 11,648-case check over a finite objective-task domain. Our results expose recurring failures between verification and settlement, where conforming work can remain unaccepted or valid evidence can be ignored. Public records and model judgments further distinguish recorded approval from evidence of task conformance, while economic analysis identifies the report, penalty, and shared-error assumptions behind these guarantees. These findings show where end-to-end guarantees fail and what must be repaired to preserve them across the workflow.
\end{abstract}

\smallskip
\section{Introduction}
\label{sec:intro}

Decentralized agent architectures connect independent agents for discovery and collaboration~\cite{tomasev2025virtual, sun2025vision,yang2026agentnet}.    Commerce protocols add agreements, evaluation, and payment for their services~\cite{virtuals-acp,crapis2026erc8183,x402,ap2,mpp-docs,okx-app}. We study \DAEs, ecosystems in which agents transact through interoperable cryptographic and economic mechanisms without a single platform controlling the full task workflow.

Consider a simple escrow-based task. An agent seeking a service, the \emph{requester}, discovers candidate agents through a registry and selects a \emph{provider} to perform the task. They agree on the requirements and price, and the requester then locks funds in escrow. The provider submits a deliverable, an \emph{evaluator} assesses it against the agreed requirements, and the escrow releases payment on approval~\cite{erc8004,virtuals-acp,crapis2026erc8183}.


Each lifecycle step may work correctly on its own: the registry authenticates identities, the escrow enforces the agreed price, and the evaluator returns a verdict. However, mechanism-level correctness does not ensure that the task meets its security and economic requirements. Registration can confirm who the provider is, but not necessarily whether it can perform the claimed task~\cite{xiong2026erc8004,mafrur2026erc8004}. Likewise, an authorized evaluator may approve a result without providing enough evidence that it meets the agreed requirements~\cite{erc8004,crapis2026erc8183}. Even when such evidence is available, the payment contract may not use it. As a result, a correct escrow can still release funds based on an approval that is only weakly tied to the actual quality of the delivered work. The problem can arise even earlier: counterparty-controlled language may influence a purchasing decision before the agent authorizes the transaction~\cite{allouah2026agentic}.


These examples point to a common issue in DAE ecosystems: a guarantee established at one stage may be missing, weakened, or ignored when a later decision depends on it. Existing work has studied several parts of this problem. Prior systematizations examine service delivery, cross-stage consistency, negotiation, and shared evaluator failures~\cite{zhang2026a2apay,mao2026sok}, while protocol and deployment studies identify specific failure modes in payment, verification, and agent infrastructure~\cite{allegrini2025formal,li2026x402attacks,ling2026freeriding,zamyatin2021sok,xiong2026erc8004,wang2026http402}. What is still missing is a task-level view of how these guarantees compose across the full workflow, and which downstream decisions must preserve and use them.


We address this gap by organizing security and economic requirements across six lifecycle stages and introducing \emph{guarantee closure}. The criterion asks whether each required guarantee is established under stated assumptions and remains available to every later decision that depends on it. This lets us distinguish missing guarantees from cases where valid evidence exists but is not used, as well as cases where a correct acceptance rule exists but cannot be reached in practice.

We evaluate this view through controlled and native workflows, public records, model judgments, and economic analysis. The workflow experiments expose composition failures and test targeted repairs. The evidence analysis separates recorded approval from evidence of task conformance, while the economic analysis identifies the report, penalty, and error assumptions needed for the corresponding guarantees.

The paper makes four contributions:
\begin{itemize}[leftmargin = *]

\item \textbf{A lifecycle requirement and mechanism map.} We define 17 stage-indexed security and economic properties, including separate notions of receipt soundness and completeness, and assess their support across $12$ systems and standards, five reusable mechanism families, and four classical baselines (\S\ref{sec:requirements} and \S\ref{sec:mechanisms}).


\item \textbf{A task-level composition criterion and evaluation.} We define \emph{guarantee closure} to test whether task-required guarantees are established and preserved across dependent downstream decisions. We apply it to two architectures across four task classes, using controlled perturbations, targeted repairs, and native implementation comparisons to expose missing guarantees, dropped cross-stage guarantees, and unavailable acceptance paths (\S\ref{sec:closure}).


\item \textbf{Evidence and economic explanations for closure failures.} We study what settlement records can and cannot establish about task conformance, trace fixed evaluator judgments through to payment outcomes, and analyze how persuasion-sensitive reports, corruption incentives, and correlated evaluator errors limit the guarantees that these workflows can support (\S\ref{sec:measurement} and \S\ref{sec:econ}).


\item \textbf{A gap-derived research agenda.} From the identified closure failures, we derive requirements for capability attestation, persuasion-robust interfaces, receipt-consuming settlement, and evaluator markets with explicit evidence, incentive, and error assumptions (\S\ref{sec:agenda}).

\end{itemize}

\section{Preliminaries and System Model}
\label{sec:preliminaries}

\subsection{Definitions}




A \emph{Decentralized Agentic Economy (\DAE)} connects autonomous participants through cryptographic and economic mechanisms without a single platform controlling the full task workflow. An \emph{\ac{an}} has a persistent identity and controls resources such as models, data, tools, memory, or compute. It can request a task, provide a service, or evaluate another agent's work. Individual stages may be brokered while still using open payment and settlement protocols.

A \emph{mechanism} implements a lifecycle function, such as an auction, escrow, or receipt scheme. A \emph{guarantee} is a property provided by a mechanism under stated assumptions. A \emph{workflow} combines mechanisms across stages. We assess claims from systems, specifications, and reusable mechanism families against their supporting documentation.

A \emph{verification receipt} records an acceptance decision together with evidence bound to the task, deliverable, and acceptance condition. Settlement \emph{consumes} the receipt when it validates this evidence and makes payment conditional on the resulting acceptance decision. Recording a receipt without using it in the payment decision leaves the verification-to-settlement dependency unenforced.

\subsection{The Task Lifecycle}
\label{sec:lifecycle}

Figure~\ref{fig:lifecycle} organizes the task lifecycle into six stages. In \textbf{S1 discovery}, the requester identifies candidate providers and considers their identity, capability, and reputation. \textbf{S2 negotiation} establishes the task specification, allocation, price, and terms. The provider produces the deliverable in \textbf{S3 execution}. In \textbf{S4 verification}, a verifier or evaluator determines whether the deliverable satisfies the agreed requirements. \textbf{S5 settlement} transfers value under the release condition, while \textbf{S6 dispute} resolves contested outcomes and updates reputation for future discovery. Roles depend on the task, so the same agent may take different roles across workflows. 

The guarantees established at each stage must be preserved and used by later decisions that depend on them. Capability claims used in provider selection require supporting evidence, and verification must apply the acceptance rule agreed during negotiation. Settlement must then condition payment on the resulting acceptance evidence for the same task and output. Even when individual mechanisms operate correctly, a later decision may omit or ignore a required guarantee. Section~\ref{sec:closure} formalizes these dependencies through guarantee closure.

\subsection{Threat Model}
\label{sec:threatmodel}

A Byzantine requester, provider, evaluator, registry operator, or relayer may deviate arbitrarily, collude, create Sybil identities, or exploit ordering and information asymmetries. Honest agents follow their prescribed protocol. Their model-generated decisions remain susceptible to misleading inputs and dependent errors. We examine two corresponding conditions.

\keypoint{A1. Persuasion-sensitive reports.} 
We consider counterparty-controlled language that changes an honest agent's bid or evaluation report while task-relevant facts remain fixed. LLM negotiation, shopping, and prompt-injection studies motivate this report channel~\cite{bianchi2024negotiation,allouah2026agentic,greshake2023,zeng2024johnny}. Section~\ref{sec:econ} states the bounded-report model and its reachability assumptions. Human-directed persuasion is a related, separately measured phenomenon~\cite{salvi2025persuasion}.

\keypoint{A2. Correlated evaluator failures.}
Agents can share models, retrieval sources, prompts, or tools and therefore make dependent errors~\cite{mao2026sok}. Byzantine fault tolerance already permits coordinated faults up to a declared number or weight of faulty participants. Shared failures matter when they exceed that budget or violate the error assumptions used to justify a guarantee. Proposition~\ref{prop:correlated} examines what committee size can establish under a specified dependence model.

The analysis starts above consensus. Each workflow declares its assumptions about principal authorization, control-flow integrity, tool execution, signer custody, authenticated delivery, bridge trust, and finality. Model-level compromise is part of the environment in which the infrastructure is assessed. Within that declared environment, A1 changes the reports generated by an otherwise fixed policy, while A2 changes the dependence between evaluator errors.

\section{Stage-Indexed Security and Economic Requirements}
\label{sec:requirements}

\begin{figure*}[t]
\centering
\includegraphics[width=\textwidth]{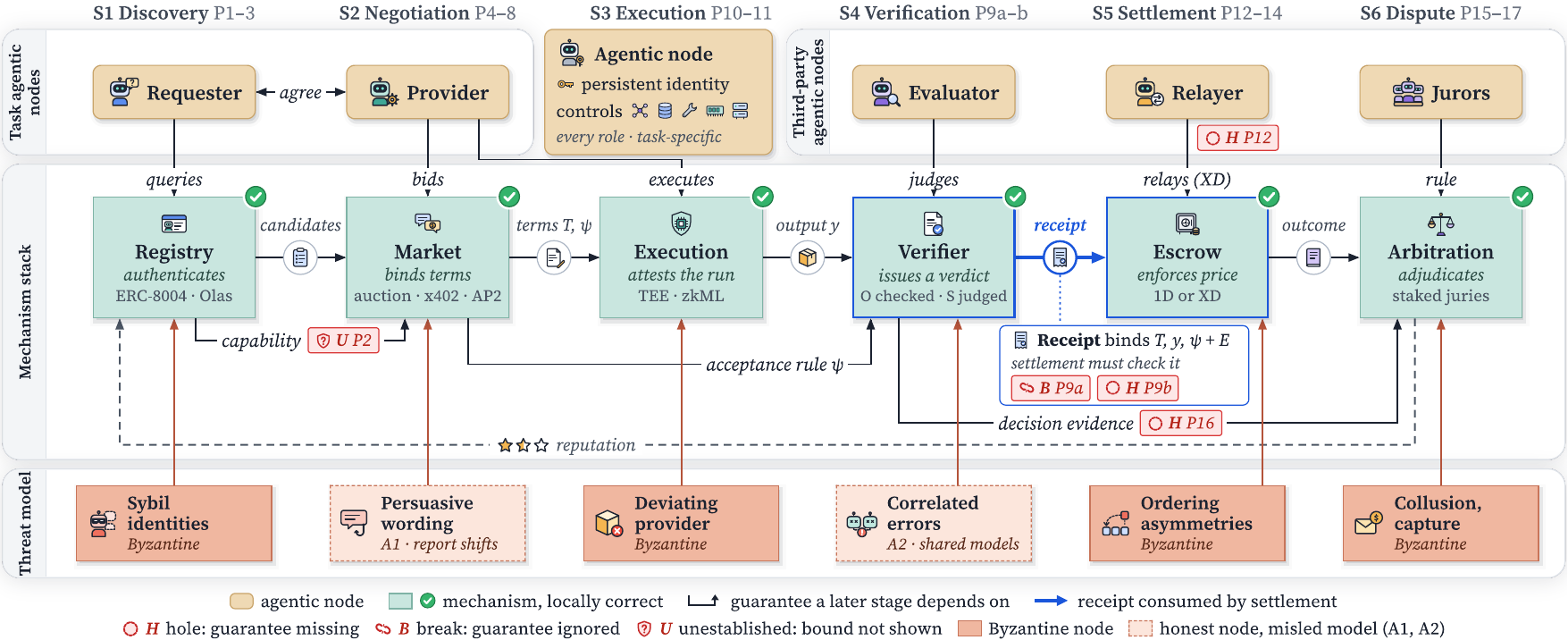}
\caption{\textbf{\DAE Task lifecycle}: agent roles, mechanisms, security and economic requirements, and threats across six stages. Negotiation fixes the acceptance predicate, verification produces a receipt binding the task, output, predicate, and supporting evidence, and settlement must validate this receipt before releasing payment. The annotated gaps illustrate how locally correct mechanisms can leave required guarantees missing, ignored, or unestablished.}
\label{fig:lifecycle}
\label{fig:sp_overview}
\end{figure*}
We classify tasks along two dimensions. \textit{Objective tasks} (O) have a checkable acceptance predicate, while \textit{subjective tasks} (S) require judgement. Verification and payment may occur within one security domain or across separate domains with distinct state, finality, or trust assumptions. These dimensions define four task classes, with 1D denoting single-domain and XD denoting cross-domain settlement.

Table~\ref{tab:props} assigns requirements $R(S_i)$ to each stage $S_i$, each with a distinct security or economic loss. Core requirements cover participation, payment, and accountability, while conditional requirements apply to features such as capability claims, confidentiality, or usage-based payment. Appendix~\ref{app:formal} formalizes the central criteria and the aggregation of closure bounds.

\subsection{Rationale for Stage-Indexed Requirements}
\label{sec:req-exec}

\keypoint{S1 Discovery.}
A requester needs to know both who the provider is and what it can do. Sybil resistance (P1) bounds influence gained through additional identities~\cite{douceur2002sybil}. Capability soundness (P2) connects advertised performance to the deployed agent on a reference task distribution. P9a instead concerns soundness for one deliverable. Reputation unforgeability and portability (P3) requires feedback whose authors and interactions remain verifiable across discovery contexts.

\keypoint{S2 Negotiation.}
An agreed price does not by itself make participation worthwhile or truthful reporting optimal. Incentive compatibility (P4) makes truthful reporting optimal under the declared solution concept~\cite{myerson1983}. Individual rationality (P5) requires expected utility, conditional on participant type, to meet the outside option. It includes execution costs, penalties, and nonpayment under the declared outcome distribution. Weak budget balance (P6) requires expected incoming transfers to cover outgoing transfers. Collusion resistance (P7) bounds additional coalition utility from coordinated deviation. Persuasion robustness (P8) instead requires fact-equivalent counterparty messages to preserve the report, or its conditional distribution for randomized policies. P8 protects report generation before commitment.

\begin{table}[t]
\centering
\begin{threeparttable}
\caption{Stage-indexed requirements by task class.}
\label{tab:props}
\scriptsize
\setlength{\tabcolsep}{3.2pt}
\renewcommand{\arraystretch}{1.02}

\begin{tabular}{@{}clp{3.8cm}cccc@{}}
\toprule
\textbf{Stage} & \textbf{P} & \textbf{Property}
& \textbf{O,1D} & \textbf{O,XD} & \textbf{S,1D} & \textbf{S,XD} \\
\midrule

\multirow{3}{*}{S1}
& P1 & Sybil resistance
& C & C & C & C \\
& P2 & Capability soundness
& C & C & C & C \\
& P3 & Unforgeable, portable reputation
& C & C & C & C \\

\cmidrule(lr){1-7}
\multirow{5}{*}{S2}
& P4 & Incentive compatibility
& C & C & C & C \\
& P5 & Individual rationality
& \fullc & \fullc & \fullc & \fullc \\
& P6 & Weak budget balance
& \fullc & \fullc & \fullc & \fullc \\
& P7 & Collusion resistance
& C & C & C & C \\
& P8 & Persuasion robustness
& C & C & C & C \\

\cmidrule(lr){1-7}
\multirow{2}{*}{S3}
& P10 & Task-scoped confidentiality
& C & C & C & C \\
& P11 & Metering integrity
& C & C & C & C \\

\cmidrule(lr){1-7}
\multirow{2}{*}{S4}
& P9a & Receipt soundness
& \fullc & \fullc & \fullc & \fullc \\
& P9b & Receipt completeness
& \fullc & \fullc & \fullc & \fullc \\

\cmidrule(lr){1-7}
\multirow{3}{*}{S5}
& P12 & Cross-domain conditional atomicity
& \emptyc & \fullc & \emptyc & \fullc \\
& P13 & Timeliness and progress
& \fullc & \fullc & \fullc & \fullc \\
& P14 & Ordering fairness / EV resistance
& C & C & C & C \\

\cmidrule(lr){1-7}
\multirow{3}{*}{S6}
& P15 & Accountability / non-repudiation
& \fullc & \fullc & \fullc & \fullc \\
& P16 & Evidence verifiability
& \fullc & \fullc & \fullc & \fullc \\
& P17 & Capture resistance
& C & C & C & C \\
\bottomrule
\end{tabular}

\begin{tablenotes}[flushleft]
\scriptsize
\item[-] \textbf{Key.} \fullc: adopted core requirement; C: conditional; \emptyc: not required by task class alone. O/S: objective/subjective; 1D/XD: single-/cross-domain settlement. 
\end{tablenotes}
\end{threeparttable}
\end{table}

\keypoint{S3 Execution.}
Execution can expose private inputs and create incentives to misreport resource use. Task-scoped confidentiality (P10) limits disclosure of requester inputs and provider resources across execution, verification, and metering to what the task authorizes. When payment depends on usage, metering integrity (P11) bounds discrepancies between recorded and actual resource use.

\keypoint{S4 Verification.}
Verification needs both a bound on accepting nonconforming work and a feasible acceptance path for conforming work. Receipt soundness (P9a) bounds acceptance of nonconforming work by a declared error bound or attack-cost floor under the predicate and trust model. Receipt completeness (P9b) gives honest providers with conforming outputs an acceptance path within the task's time and cost bounds. A sound verifier may still be unavailable or unaffordable. Subjective soundness additionally depends on evaluator selection, incentives, and error.

\keypoint{S5 Settlement.}
A correct verification result helps only if it controls payment for the same task and output. Across domains, conditional atomicity (P12) requires matching finalized acceptance, progress under delivery/finality bounds, and consistent expiry/refund rules~\cite{zamyatin2021sok}. Timeliness and progress (P13) requires completion within the horizon and, where needed, settlement of verifiable partial results. Ordering fairness and extractable-value resistance (P14) bounds gains from controlling ordering, messages, or context~\cite{daian2020flash,kelkar2020order,li2025transaction}. P8 instead asks whether that control changes a report.

\keypoint{S6 Dispute and Governance.}
A dispute needs evidence of who acted and how the decision was reached. Accountability and non-repudiation (P15) attributes consequential actions to identities using admissible evidence that resists denial. Evidence verifiability (P16) protects provenance and permits replay of deterministic steps. Subjective judgments also require evidence about the work. Capture resistance (P17) prevents coalitions within the stated resource bound from durably controlling adjudication or rules. Properties can recur after their primary stage. For example, P7 also applies to coalition-sensitive evaluator voting.

\insightbox{\textbf{Acceptance guarantees are not interchangeable.} Replaying deterministic decision steps (P16) does not by itself establish soundness (P9a). Soundness alone does not ensure that honest providers with conforming outputs obtain acceptance within the task's time and cost bounds (P9b). Settlement must condition payment on validated, task-bound acceptance, not merely record a verdict.}

\label{sec:taskclass}
The concrete workflow activates the conditional entries in Table~\ref{tab:props}. For an objective computation at a fixed price, P9a and P9b require a sound acceptance decision and a feasible route to obtain it. Charging by resource use additionally activates P11. Moving payment to another domain activates P12 without changing the acceptance rule. Selecting the provider from a claimed benchmark activates P2 before execution, even when the final output can be checked independently. This task-specific requirement set governs both the mechanism assessment and the closure test.

\section{Mechanisms and Provided Guarantees}
\label{sec:mechanisms}

We compare $12$ systems and standards, five reusable mechanism families, and four classical baselines in Table~\ref{tab:property-matrix}. We select this corpus to cover public designs for open agent tasks, including the A402 research prototype. We exclude wallet-only rails and closed platforms. We provide the inclusion criteria, scoring rules, and evidence behind the core acceptance ratings in Appendix~\ref{app:corpus}.

Each assessment identifies a subject, interface, and acceptance condition. Its evidence may come from a specification, implementation, observation, or proposal. Full support meets the criterion under stated cryptographic and base-ledger primitives. Restricted support covers a subclaim or requires additional trust, resource, or timing conditions. It can satisfy a task under admissible assumptions. \emptyc\ denotes an applicable but unestablished requirement. {\outc{} denotes inapplicability.} Family rows summarize candidate guarantees from cited members. A workflow must check that the selected members and their assumptions are jointly compatible.

An interface's role is not the same as a task's needs. AP2, x402, and MPP authorize payment without establishing work conformance, so P9a is unsupported. They supply no acceptance service, making P9b inapplicable. Conversely, a ZK proof can establish P9a for a sorting relation without supplying the complete decision history needed for P16.

ACP and ERC-8183 both receive unsupported P9b because their evaluator-controlled acceptance lacks an established bounded alternative during refusal. Restricted support requires an identified progress mechanism with explicit conditions. Table~\ref{tab:acceptance-evidence} records the $12$ systems' P9a, P9b, and P16 evidence.

Table~\ref{tab:support-summary} separates system and family counts, describing support within the selected corpus. Full P9a is relation-specific, whereas full P16 covers the task's decision evidence. Accepting a hardware or challenger assumption may make restricted support sufficient without changing the strict classification.

\begin{table}[t]
\centering
\begin{threeparttable}
\caption{Support counts by comparison unit.}
\label{tab:support-summary}
\setlength{\tabcolsep}{3pt}
\begin{tabular}{>{\centering\arraybackslash}p{3cm}>{\centering\arraybackslash}p{1cm}rrrr}
\toprule
\textbf{Comparison unit} & \multicolumn{1}{c}{\textbf{Property}} & F & R & U & {\outc{}} \\
\midrule
Systems/standards (12) & P9a & 0 & 7 & 5 & 0 \\
Systems/standards (12) & P9b & 0 & {4} & {4} & 4 \\
Systems/standards (12) & P16 & 0 & 8 & 4 & 0 \\
Mechanism families (5) & P9a & 1 & 3 & 1 & 0 \\
Mechanism families (5) & P9b & 0 & 4 & 0 & 1 \\
Mechanism families (5) & P16 & 0 & 5 & 0 & 0 \\
\bottomrule
\end{tabular}
\begin{tablenotes}[flushleft]
\scriptsize
\item[-] F/R/U: full/restricted/unsupported; \outc: inapplicable. Counts cover the selected corpus. P9a/P9b denote receipt soundness/completeness; P16 denotes evidence verifiability.
\end{tablenotes}
\end{threeparttable}
\end{table}

\begin{table*}[htbp]
\centering
\begin{threeparttable}
\caption{Mechanism support for the lifecycle properties.}
\label{tab:property-matrix}
\scriptsize
\setlength{\tabcolsep}{1.45pt}
\renewcommand{\arraystretch}{1.02}
\newcommand{\propertyrot}[1]{%
  \rotatebox[origin=c]{90}{\makebox[1.65cm][l]{\scriptsize\strut #1}}%
}
\setlength{\bigstrutjot}{\dimexpr1.65cm-\baselineskip+\aboverulesep+\belowrulesep+\cmidrulewidth\relax}
\newcommand{\matrixhead}[1]{\multirow{2}[1]{*}[10pt]{#1}}
\begin{tabular*}{\textwidth}{@{\extracolsep{\fill}}>{\centering\arraybackslash}p{0.8cm} >{\raggedleft\arraybackslash}p{4.0cm} >{\raggedleft\arraybackslash}p{2.7cm}
*{3}{C{0.46cm}} *{5}{C{0.46cm}} *{4}{C{0.46cm}} *{3}{C{0.46cm}} *{3}{C{0.46cm}}@{}}
\toprule
\multicolumn{1}{c}{\matrixhead{\rot{\textbf{Group}}}}
& \multicolumn{1}{c}{\matrixhead{\textbf{Artifact / Mechanism family}}}
& \multicolumn{1}{c}{\matrixhead{\textbf{Class}}}
& \propertyrot{P1 Sybil} & \propertyrot{P2 Capability} & \propertyrot{P3 Reputation}
& \propertyrot{P4 IC} & \propertyrot{P5 IR} & \propertyrot{P6 BB} & \propertyrot{P7 Collusion} & \propertyrot{P8 Persuasion}
& \propertyrot{P9a Sound} & \propertyrot{P9b Complete} & \propertyrot{P10 Confid.} & \propertyrot{P11 Metering}
& \propertyrot{P12 Atomicity} & \propertyrot{P13 Timeliness} & \propertyrot{P14 Ordering}
& \propertyrot{P15 Account.} & \propertyrot{P16 Evidence} & \propertyrot{P17 Capture} \\
\cmidrule(lr){4-6}\cmidrule(lr){7-11}\cmidrule(lr){12-15}\cmidrule(lr){16-18}\cmidrule(lr){19-21}
& & & \multicolumn{3}{c}{\textbf{S1 Discovery}}
& \multicolumn{5}{c}{\textbf{S2 Negotiation}}
& \multicolumn{4}{c}{\textbf{S3 and S4 Exec./Verify}}
& \multicolumn{3}{c}{\textbf{S5 Settlement}}
& \multicolumn{3}{c}{\textbf{S6 Dispute}} \\
\midrule
\multirow{12}{*}[-1.5\baselineskip]{\rotatebox[origin=c]{90}{\textit{Systems and standards (12)}}}
& ERC-8004 draft \cite{erc8004} & Identity/reputation registry
& \emptyc & \emptyc & \halfc & \outc & \outc & \outc & \emptyc & \outc & \halfc & \outc & \outc & \outc & \outc & \outc & \outc & \halfc & \halfc & \emptyc \\
& ACP / Virtuals \cite{virtuals-acp} & Escrow + evaluator
& \emptyc & \emptyc & \halfc & \emptyc & \halfc & \fullc & \emptyc & \emptyc & \halfc & {\emptyc} & \emptyc & \outc & \outc & \halfc & \emptyc & \halfc & \halfc & \emptyc \\
& ERC-8183 draft \cite{crapis2026erc8183} & Job state machine
& \outc & \outc & \halfc & \emptyc & \halfc & \fullc & \emptyc & \emptyc & \halfc & \emptyc & \halfc & \outc & \outc & \halfc & \halfc & \halfc & \halfc & \emptyc \\
& AP2 v0.2 \cite{ap2} & Verifiable payment credentials
& \outc & \outc & \outc & \outc & \halfc & \fullc & \emptyc & \emptyc & \emptyc & \outc & \halfc & \outc & \outc & \halfc & \halfc & \halfc & \halfc & \outc \\
& x402 v2 \cite{x402,x402-spec-v2,x402-bazaar,x402-offer-receipt} & Web posted-price payments
& \outc & \emptyc & \outc & \outc & \halfc & \fullc & \emptyc & \emptyc & \emptyc & \outc & \outc & \halfc & \outc & \halfc & \halfc & \halfc & \halfc & \outc \\
& A402 prototype \cite{li2026a402} & Atomic service channels
& \emptyc & \halfc & \outc & \outc & \halfc & \halfc & \emptyc & \emptyc & \halfc & \halfc & \halfc & \outc & \halfc & \halfc & \emptyc & \halfc & \halfc & \outc \\
& MPP \cite{mpp-docs,httpauth-payment} & Payment authorization
& \outc & \emptyc & \outc & \outc & \halfc & \fullc & \emptyc & \emptyc & \emptyc & \outc & \halfc & \halfc & \outc & \halfc & \halfc & \halfc & \emptyc & \outc \\
& OKX APP proposal \cite{okx-app,okx-app-whitepaper} & Brokered payment intents
& \outc & \outc & \outc & \emptyc & \halfc & \fullc & \emptyc & \emptyc & \halfc & \halfc & \emptyc & \halfc & \outc & \halfc & \halfc & \halfc & \halfc & \emptyc \\
& Bittensor \cite{bittensor2021,bittensor-docs} & Staked subjective scoring
& \halfc & \emptyc & \halfc & \emptyc & \halfc & \emptyc & \emptyc & \emptyc & \halfc & \halfc & \emptyc & \emptyc & \outc & \halfc & \halfc & \halfc & \emptyc & \emptyc \\
& Autonolas / Olas \cite{autonolas,olas-mech} & Bonded service registry
& \halfc & \emptyc & \halfc & \outc & \halfc & \halfc & \emptyc & \emptyc & \halfc & \halfc & \emptyc & \halfc & \outc & \halfc & \emptyc & \halfc & \halfc & \emptyc \\
& Fetch.ai \cite{fetchai,fetchai-almanac,uagents} & Agent marketplace
& \halfc & \emptyc & \emptyc & \outc & \emptyc & \fullc & \emptyc & \emptyc & \emptyc & \emptyc & \emptyc & \emptyc & \outc & \halfc & \emptyc & \halfc & \emptyc & \emptyc \\
& SingularityNET \cite{singularitynet2019} & Registry + payment channels
& \emptyc & \emptyc & \emptyc & \outc & \halfc & \fullc & \emptyc & \emptyc & \emptyc & \emptyc & \emptyc & \halfc & \outc & \halfc & \halfc & \halfc & \emptyc & \emptyc \\
\midrule
\multirow{4}{*}[-1.5\baselineskip]{\rotatebox[origin=c]{90}{\shortstack{\textit{Classical}\\\textit{baselines}\\\textit{(4)}}}}
& Contract Net \cite{smith1980contractnet} & Task announcement protocol
& \outc & \emptyc & \outc & \emptyc & \emptyc & \emptyc & \emptyc & \emptyc & \outc & \outc & \outc & \outc & \outc & \halfc & \outc & \emptyc & \emptyc & \outc \\
& False-name-proof auctions \cite{yokoo2001aij,yokoo2004geb} & Sybil-robust mechanism
& \fullc & \outc & \outc & \fullc & \fullc & \fullc & \halfc & \emptyc & \outc & \outc & \outc & \outc & \outc & \outc & \outc & \outc & \outc & \outc \\
& Faithful distributed implementation \cite{shneidman2004faithful,feigenbaum2002bgp} & Mechanism + node incentives
& \outc & \outc & \outc & \fullc & \halfc & \halfc & \halfc & \emptyc & \halfc & \outc & \outc & \outc & \outc & \outc & \outc & \halfc & \halfc & \outc \\
& Computational trust / reputation \cite{huynh2006fire,teacy2006travos,kamvar2003eigentrust,sabater2001regret,cheng2005sybilproof} & Witness-based reputation
& \halfc & \halfc & \halfc & \outc & \outc & \outc & \emptyc & \outc & \outc & \outc & \outc & \outc & \outc & \outc & \outc & \outc & \outc & \outc \\
\midrule
\multirow{5}{*}[-0.5\baselineskip]{\rotatebox[origin=c]{90}{\shortstack{\textit{Reusable}\\\textit{mechanism}\\\textit{families (5)}}}}
& zkML receipts \cite{kang2022zkml,sun2024zkllm,ezkl,qu2025zkgpt} & ZK verification
& \outc & \halfc & \outc & \outc & \outc & \outc & \outc & \outc & \fullc & \halfc & \fullc & \emptyc & \outc & \halfc & \outc & \halfc & \halfc & \outc \\
& TEE-based inference \cite{costan2016sgx,zhu2024hoppercc,knauth2018ratls} & Hardware attestation
& \outc & \halfc & \outc & \outc & \outc & \outc & \outc & \outc & \halfc & \halfc & \halfc & \emptyc & \outc & \halfc & \outc & \halfc & \halfc & \outc \\
& Optimistic receipts \cite{kalodner2018arbitrum,teutsch2017truebit,conway2024opml,lee2026challenge} & Fraud proofs
& \outc & \emptyc & \outc & \outc & \halfc & \outc & \halfc & \outc & \halfc & \halfc & \halfc & \emptyc & \outc & \halfc & \halfc & \halfc & \halfc & \outc \\
& Intent-based cross-domain settlement \cite{across-intent-lifecycle,erc7683,zamyatin2021sok,across-spokepool,umip-157,umip-179,across-root-bundles} & Solver fill and repayment
& \outc & \outc & \outc & \outc & \outc & \outc & \outc & \outc & \emptyc & \outc & \emptyc & \outc & \emptyc & \halfc & \halfc & \halfc & \halfc & \outc \\
& Staked arbitration \cite{kleros2019,uma-oracle,realityeth} & Juried / bonded disputes
& \fullc & \outc & \outc & \outc & \halfc & \halfc & \emptyc & \emptyc & \halfc & \halfc & \emptyc & \outc & \outc & \halfc & \outc & \halfc & \halfc & \halfc \\
\bottomrule
\end{tabular*}
\begin{tablenotes}[flushleft]
\scriptsize
\item[-] \textbf{Symbols.} \fullc\ denotes full support, \halfc\ restricted support, \emptyc\ unsupported guarantees, and \outc\ inapplicable properties. Strict-support counts include only full support.
\item[-] \textbf{Our assessment.} The matrix includes twelve systems and standards, five reusable mechanism families, and four classical baselines. Family rows summarize claims from the cited members and require compatible assumptions and interfaces when mechanisms are combined. Appendix~\ref{app:corpus} gives the assessment criteria.
\item[-] \textbf{Abbreviations.} ACP denotes the Agent Commerce Protocol. AP2 and OKX APP denote the respective Agent Payments Protocols. MPP denotes the Machine Payments Protocol. ERC denotes Ethereum Request for Comments. IC, IR, and BB denote incentive compatibility, individual rationality, and budget balance.
\end{tablenotes}
\end{threeparttable}
\end{table*}

\subsection{Identity, Attestation, and Reputation}
\label{sec:mech-identity}

Registries and credentials bind identities to keys or endpoints, but identity binding alone does not limit the influence gained from creating additional identities~\cite{w3c-did,w3c-vc,erc8004}. Kleros addresses this problem through stake-weighted juror selection, so splitting a fixed stake across multiple identities does not increase its aggregate selection weight~\cite{kleros2019}. Bittensor also weights influence by stake, but stake-ranked permits restrict participation. Its P1 support therefore covers only the admitted validator set~\cite{bittensor-validating}.

Capability evidence must connect evaluation to the serving configuration. {Attested evaluation} already binds {evaluation results} to online model responses under TEE assumptions~\cite{schnabl2025attestable}. Attestation and weight commitments support related bindings~\cite{costan2016sgx,knauth2018ratls,sun2024zkllm}. Full P2 additionally requires the claimed statistical performance, complete agent configuration, and validity interval to match the deployed service. These bindings define the task-specific capability claim to be established.

Reputation adds an interaction-level requirement. ERC-8004 authenticates feedback authors and permits record reuse, but evidence linking feedback to completed work is optional, giving restricted P3~\cite{erc8004}. FIRE authenticates certified references under PKI and signature assumptions, while their interaction claims still rely on honest reporting. Rating correctness is a separate P9a obligation~\cite{huynh2006fire}. TRAVOS supplies statistical reliability estimates under stationary behavior, supporting a P2 subclaim without binding those estimates to the serving configuration~\cite{teacy2006travos}.

\subsection{Negotiation and Market Mechanisms}

\label{sec:mech-market}

Payment authorization protects agreed terms after a decision has been made, but it does not protect the process that produced that decision. x402 and the Machine Payments Protocol (MPP) authorize payments under committed terms, while signed agreements and intents bind negotiated terms~\cite{x402,mpp-docs,httpauth-payment,virtuals-acp,ap2,erc7521}. P8 concerns the preceding report-generation process. A persuasive listing can change a report before commitment, so authorization and signature checks do not establish P8. Purchasing and price-negotiation experiments show sensitivity to listings and anchors~\cite{allouah2026agentic,takenami2025anchoring}. Prompt injection changes ranking in an AP2 shopping prototype~\cite{debi2026ap2redteam}. Auction experiments separately show how reasoning prompts and strategy advice change bids~\cite{shah2025auction}.

Classical mechanisms show what stronger incentive guarantees require. Leveled Division Set (LDS) auctions provide P1 and P4 within a fixed division and reservation-price model. The Vickrey--Clarke--Groves (VCG) result requires both true and admissible reported types to lie in a valuation domain with a concave surplus function. Faithful distributed implementation provides P4 at an ex post Nash equilibrium. Their restricted P7 results cover one economic participant, including its use of multiple identities or deviations in reporting and computation. Under the stated assumptions, that participant cannot increase its utility through the covered deviations. These results do not extend to collusion among distinct participants without an additional coalition bound~\cite{yokoo2001aij,yokoo2004geb,shneidman2004faithful}. Appendix~\ref{app:rationale} gives the corresponding conditions and bounds.

Current agent-market mechanisms provide weaker guarantees at this interface. Bittensor bounds the reward share retained by honest stake under its stated stake and scoring conditions, but it does not establish truthful reporting as an optimal strategy or an additional bound on coalition utility~\cite{bittensor-docs}. Signatures and multisigs similarly constrain authorization without bounding gains from coordinated deviation. P7 therefore requires a coalition-utility bound, while P8 requires control over the report-generation process before authorization.

\subsection{Verifiable Execution}

\label{sec:mech-verify}

Verification mechanisms differ in both what they establish and what they assume. Zero-knowledge (ZK) receipts provide soundness for a specified relation. Trusted execution environments (TEEs) attest execution under a measured configuration and hardware root of trust. Optimistic receipts instead rely on challengers who can contest invalid claims~\cite{groth2016,sun2024zkllm,conway2024opml,costan2016sgx}. These mechanisms can support P9a, but proof-generation cost, attestation availability, and challenge windows also determine whether an honest provider can obtain acceptance within the task's time and cost bounds under P9b and P13~\cite{qu2025zkgpt,vanschaik2024sgxfail}.

The acceptance claim is only as strong as the statement being verified. Proof of execution establishes task conformance only when the proved relation encodes the agreed acceptance condition. A TEE receipt must likewise bind the delivered result to the measured configuration, while optimistic acceptance retains its challenger and timing assumptions. Settlement must therefore validate the acceptance claim together with its task binding and supporting assumptions.

These execution guarantees do not by themselves establish P11. The cited zkML and TEE mechanisms provide computation proofs or configuration attestation. Benchmark timings do not authenticate task-bound consumption~\cite{sun2024zkllm,knauth2018ratls,zhu2024hoppercc}. Bittensor's utility scores likewise do not supply such measurements~\cite{bittensor2021,bittensor-docs}. P11 remains unsupported for these assessed sources. A usage-based workflow needs separate evidence binding measured consumption to its task and interval.

Subjective work shifts this dependency from a technical verifier to an evaluator. The Agent Commerce Protocol (ACP) uses evaluator judgment, and ERC-8183 permits a client-selected, unbonded evaluator~\cite{virtuals-acp,crapis2026erc8183}. Hence, their P9a support depends on the evaluator and the evidence supporting its judgment. P9b introduces a separate requirement. As ERC-8183 refunds the client after expiry, a conforming provider depends on an evaluator response to obtain acceptance and payment. Table~\ref{tab:verif-families} (Appendix~\ref{app:corpus}) compares these verification paths, and \S\ref{sec:econ} analyzes evaluator incentives and error dependence.

A402 couples execution, payment, and result delivery through TEE-assisted adaptor signatures~\cite{li2026a402}. This is an existing integrated binding mechanism. Its guarantee depends on trusted execution and the chosen service logic. Establishing semantic task acceptance and bounded completion requires corresponding acceptance and progress guarantees.

\subsection{Settlement and Cross-Domain Coordination}

\label{sec:mech-settle}

Settlement inherits the condition on which payment is released. An escrow can therefore execute correctly while paying against a condition that does not establish task acceptance~\cite{virtuals-acp,crapis2026erc8183}. The problem becomes more explicit across security domains, where finalized acceptance in one domain must remain bound to payment in another. The selected Across integration validates destination funds delivery before origin repayment, while ERC-7683 assigns settlement security to its implementing protocol~\cite{across-intent-lifecycle,erc7683}. Funds delivery alone does not establish task acceptance.

The assessed Across integration supports application-level acceptance checks, which our funds-delivery baseline does not exercise~\cite{across-spokepool,umip-179}. P12 requires payment to consume finalized acceptance of the same task and result, with coordinated expiry and refund conditions. Appendix~\ref{app:rationale} summarizes the mechanism-specific assumptions. Verification-conditioned payment provides related constructions~\cite{xu2021htlc,campanelli2017zkcp,fuchsbauer2019wi}.

Ordering introduces a separate settlement risk. Maximal extractable value (MEV) research measures ordering-based extraction~\cite{daian2020flash,li2025transaction} and develops ordering defenses under declared network assumptions~\cite{kelkar2020order,heimbach2022sok}. The optimistic family supplies a restricted P14 subclaim. When all valid challenges are included and priority fees are zero, equal rewards do not change under permutations of that fixed set~\cite{lee2026challenge}. This does not cover exclusion, variable fees, or message manipulation. Agent markets additionally expose the messages and context used before a transaction is formed. Control over these inputs can change the economic decision itself. \S\ref{sec:econ} formalizes these additional P14 channels and their relation to P8.

\subsection{Dispute Resolution and Governance}

\label{sec:mech-govern}

Dispute mechanisms require both attribution and evidence about the disputed decision. Staked juries and optimistic oracles provide adjudication under participation and capture assumptions~\cite{kleros2019,uma-oracle,realityeth}. Job records, signed authorizations, and activity records identify who performed parts of the lifecycle~\cite{crapis2026erc8183,ap2,bittensor-docs,autonolas}. Full P15 extends this attribution to every consequential lifecycle action.

P16 adds a stronger requirement. Reconstructing a decision requires integrity-protected inputs, procedure versions, and the evidence used to reach the outcome, so task and deliverable references provide only part of the required record. Governance then adds P17, whose strength depends on stake concentration, arbitration thresholds, appeal rules, and final authority.

\gapbox{\textbf{Local guarantees do not establish an end-to-end acceptance chain.} A component can satisfy its local specification while the task still lacks an adequate acceptance chain. Downstream decisions must preserve the task binding of upstream guarantees and use them under compatible assumptions.}

\section{End-to-End Guarantee Closure}
\label{sec:closure}

A sorting task pays for an ordered output containing exactly the elements of a committed input. Even a correct verifier cannot protect payment released on authorization alone. Conversely, requiring a receipt is insufficient if only an unavailable evaluator can produce it, leaving conforming work unpaid.

Removing verification and all equivalent acceptance paths creates a missing guarantee, or \emph{hole}. Retaining adequate evidence but ignoring it at payment creates a \emph{break}. An unavailable sole evaluator creates a P9b hole because the provider has no way to obtain the required receipt.

Guarantee closure checks provision and preservation for every active task requirement. The definition records the task bindings, evidence, assumptions, validity intervals, and bounds needed to justify dependent decisions.

\begin{definition}[Guarantee Closure]
\label{def:closure}
{For a task class and its declared path conditions, let $\mathcal P=\{P1,\ldots,P8,P9a,P9b,P10,\ldots,P17\}$ and $R(S_i)\subseteq\mathcal P$ denote the requirements at stage $S_i$. For each $P\in R(S_i)$, let $Q_{P,i}(\theta)$ determine whether the claimed bound meets the task's threshold. Qualitative claims must meet their stated assessment criterion. The task fixes these predicates, permitted participant actions, admissible adversarial histories, and a dependency relation $\mathcal D_T$ before mechanisms are inspected. An edge $(P,i,j)\in\mathcal D_T$ requires the decision at $S_j$ to preserve the claim supplied for $P$ at $S_i$.}

To track whether the same guarantee is preserved across stage boundaries, we denote each guarantee as a typed record
\[
G=(P,s,T,y,\psi,E,A,I,\theta,\mathsf{prod}),
\]
where $P$ is the property, $s$ the subject, $T$ the task or state digest, $y$ the result or deliverable digest, $\psi$ the acceptance predicate, $E$ the evidence type and payload, $A$ the trust assumptions, $I$ the validity interval, $\theta$ the quantitative bound, and $\mathsf{prod}$ the producer. The bound $\theta$ carries its guarantee kind and units, such as error probability, economic cost, or time. Evidence may be a runtime record or a static certificate for a mechanism property. A schema may use $y=\bot$, denoting an absent result, before a result exists or when it is irrelevant to the claim. A later result-dependent claim must bind the actual result before consumption. A mechanism provides $P$ at $S_i$ when a valid $G$ is available with $Q_{P,i}(\theta)=1$ under $A$.

{A path assigns mechanisms or external processes to stages under environmental assumptions $\Gamma$. Each assumption must be evidenced or explicitly included in $\Gamma$, jointly valid without circular justification, and valid conditional on admitted adaptive histories. A decision consumes $G$ when its validated claim constrains the decision. Runtime consumption authenticates producer and evidence, matches the required fields among $P,s,T,y,\psi$, checks $I$ and $Q_{P,i}(\theta)$, and checks available evidence for $A$ and compatibility with $\Gamma$. Static consumption requires a transition argument preserving the bound under the mechanism, task, and history assumptions.}

A path preserves a required property when it is provided under admissible assumptions and every edge $(P,i,j)\in\mathcal D_T$ is realized by consuming the corresponding guarantee at $S_j$. Revalidation at $S_j$ may supply equivalent evidence if it establishes the same bound, bindings, and validity under admissible assumptions. At the terminal stage S6, preservation reduces to provision because there is no downstream stage. A path achieves guarantee closure relative to $\Gamma$ for the declared task profile when every property in $R(S_i)$ is preserved, with its specified bound, at each required stage.

\end{definition}

A closure claim is relative to the task's thresholds and a shared execution. When several probability bounds contribute to the claim, their sum must fit the task's error budget, including the applicable histories for adaptive calls. Time and economic bounds retain their respective units. Appendix~\ref{app:formal} gives the probability aggregation rule.

{An obligation is \emph{unestablished} when its required bound lacks support and no failure has been established under the same conditions. A \emph{hole} identifies an established absence of a source guarantee or recovery path. A \emph{break} identifies an established upstream guarantee discarded by a dependent decision. Counterexamples retain the task, assumptions, and strategy quantified by the claim. In particular, failure of one invocation policy can refute that policy's progress guarantee, while a strategy-based availability claim requires analysis of all permitted provider alternatives. Closure requires every active obligation to be established.}

\keypoint{Closure assessment.}
Fix task predicates, thresholds, allowed histories, and dependencies before assessing mechanisms. Record each guarantee's producer, evidence, assumptions, and consumer checks. Runtime claims require authenticated records, while static utility and availability claims require transition arguments over admitted histories. These obligations distinguish missing guarantees, discarded evidence, and unproved bounds. Appendix~\ref{app:closure-ledger} gives the assessment.

\subsection{Closure on a Reference Task}

With ideal commitment and authentication, the reference sorting task prefunds 10 payment units. The provider's total cost satisfies $c_{\rm P}\leq10$, the requester's net value satisfies $v_{\rm R}\geq10$, and both outside options are zero. Hence
\[
u_{\rm P}=10-c_{\rm P}\geq0,
\qquad
u_{\rm R}=v_{\rm R}-10\geq0,
\]
establishing P5. Payment cannot exceed the prefunded amount (P6). Under the declared execution and delivery conditions, a sound, callable verifier accepts by time 12, and payment completes by time 16, satisfying P9a, P9b, and P13. Signed actions and replayable records establish terminal P15 and P16.

{The reference includes an objective acceptance service that returns matching evidence within the declared response and cost bounds. Together with honest delivery, this service contract supplies the transition argument for the seven obligations. The controlled evaluation admits evaluator silence and malformed evidence, and assesses its policies under that broader outcome space. Appendix~\ref{app:reference-contract} states the reference assumptions and Table~\ref{tab:closure-ledger} gives the dependencies.}

\subsection{Closure Evaluation.}
\label{sec:closure-method}

{The evaluation uses matched tasks, predicates, and schedules to distinguish interface failures and repair effects. Native implementations measure local transfers and application outcomes. The unit of comparison is the same condition under different policies. Section~\ref{sec:measurement} separately examines the information supplied by retained task records.}

\keypoint{Controlled workflows.}
Two architectures combine objective or subjective verification with single-domain or cross-domain payment. All use known counterparties and a public fixed price. The eight configurations yield 68 property assessments. In the evaluated subjective profiles, capability-based allocation and coalition-sensitive evaluation additionally activate P2 and P7. Cross-domain payment activates P12, requiring payment to remain tied to task acceptance.

Payment-first releases on payment authorization, whereas verification-first also requires a valid evaluator receipt within one domain. Both cross-domain baselines check funds delivery rather than task acceptance. Four repairs target this path. C checks the complete receipt at payment. A checks it at source and issues a task-bound payment authorization. V makes objective verification provider-callable. E records decision inputs and procedure for replay. C and A share the same authenticated channel, source-finality assumption, and delivery bound. Appendix~\ref{app:repair-design} specifies their checks.

Crossing $15$ conditions with seven policies and eight configurations gives $840$ executions for matched comparisons. Sorting provides the objective acceptance predicate. Subjective simulations vary assigned verdicts against fixed conformance labels, isolating evidence handling rather than model quality. Appendix~\ref{app:closure-eval} gives the controls and boundary tests.

\keypoint{Native implementations.}
We execute the x402 reference implementation and unchanged ERC-8183 escrow on local Ethereum Virtual Machines (EVMs)~\cite{x402,crapis2026erc8183}, measuring transfers, application effects, and final balances. Three payment policies are compared across eleven x402 conditions, and three evaluator policies across ten ERC-8183 conditions for each of four SHA-256 document tasks. Evaluator extensions follow the reference escrow's submission and terminal-state semantics. Table~\ref{tab:native-results} lists the policies, and Appendix~\ref{app:experiments} gives the implementation settings and transaction ordering.

\subsection{Closure Findings}
\label{sec:closure-procedure}
\label{sec:closure-findings}

\begin{table}[!]
\centering
\begin{threeparttable}
\caption{Selected baseline closure outcomes.}
\label{tab:closure}
\scriptsize
\setlength{\tabcolsep}{2.6pt}
\renewcommand{\arraystretch}{1.08}
\begin{tabular*}{\columnwidth}{@{\extracolsep{\fill}}ccp{4.5cm}@{}}
\toprule
\textbf{Architecture} & \textbf{Task class} & \multicolumn{1}{c}{\textbf{Closure status}} \\
\midrule
Payment-first
& O,1D
& $H_{16}$, $B_{9a}^{4\rightarrow5}$, {$U_{9b}$} \\

Payment-first
& O,XD
& $H_{16}$, $B_{9a}^{4\rightarrow5}$, {$U_{9b}$}, $H_{12}^{\psi}$ \\

Payment-first
& S,1D
& $H_{9b}$, $H_{16}$, $U_2$, $U_7$, $U_{9a}$ \\

Payment-first
& S,XD
& $H_{9b}$, $H_{16}$, $H_{12}^{\psi}$, $U_2$, $U_7$, $U_{9a}$ \\

\addlinespace[2pt]

Verification-first
& O,1D
& $U_{9b}$, $H_{16}$, {$H_{13}$} \\

Verification-first
& O,XD
& $B_{9a}^{4\rightarrow5}$, $U_{9b}$, $H_{16}$, $H_{12}^{\psi}$ \\

Verification-first
& S,1D
& $H_{9b}$, $H_{16}$, {$H_{13}$}, $U_2$, $U_7$, $U_{9a}$ \\

Verification-first
& S,XD
& $H_{9b}$, $H_{16}$, $H_{12}^{\psi}$, $U_2$, $U_7$, $U_{9a}$ \\
\bottomrule
\end{tabular*}
\begin{tablenotes}[flushleft]
\scriptsize
\item[-] \textbf{Notation.} $H$: hole; $B$: break; $U$: unestablished obligation. Subscripts identify properties; $i\!\rightarrow\!j$ marks a guarantee not consumed from stage $i$ to $j$. $\psi$ marks task-acceptance binding for P12. Task classes follow Table~\ref{tab:props}.
\end{tablenotes}
\end{threeparttable}
\end{table}

Table~\ref{tab:closure} distinguishes the baseline closure failures. For objective single-domain tasks, payment-first has a P9a break because valid acceptance evidence does not control settlement. In subjective tasks, both architectures have a P9b hole because only the designated evaluator can issue acceptance and may remain silent. Objective P9b remains unestablished. Cross-domain settlement adds a P12 hole because the baseline adapter checks funds delivery without matching task acceptance. The evidence schemas also leave a P16 hole in every configuration, while subjective configurations additionally leave P2, P7, and P9a unestablished. Removing all acceptance paths creates P9a and P9b holes; removing only the payment check creates a P9a break while preserving P9b.

{C and A make settlement consume task-bound acceptance and prevent the same nine unbound payments in each objective cross-domain configuration. V recovers acceptance after silence or contestable rejection, but its invocation rule leaves five of eleven conforming objective cases without usable acceptance. E makes intact decision records replayable, while altered records require a recovery path beyond integrity checking. These repairs address distinct dependencies without establishing closure. Of 68 property assessments, 16 are preserved, 40 remain unestablished, and twelve have holes, comprising eight P13 holes and four subjective P9b holes.}

Recovery must cover every unusable receipt, including one carrying approval. With source acceptance at time $r$ and delivery delay $d$, the schedule requires $r+d\leq12$ and $r+d+1\leq16$. Admitting $d=10$ requires recovery by $r=2$, within the evaluation step. Later recovery requires later receipt and payment deadlines. P9b additionally requires permitted provider retries and their total cost to fit the task budget.

Finite-input enumeration covers all 364 arrays of length zero through five over $\{-1,0,1\}$, giving 11,648 additional objective executions. The outcome pattern remains unchanged for each fixed configuration and perturbation. 
Appendix~\ref{app:additional-outcomes} reports the acceptance misses that remain after repair, while Appendices~\ref{app:closure-sensitivity} and~\ref{app:finite-stress} report the sensitivity and robustness checks.

\subsection{Native Workflow Evaluation}
\label{sec:whynot}
\label{sec:closure-results}


\begin{table}[b]
\centering
\begin{threeparttable}
\caption{Native x402 and ERC-8183 outcomes.}
\label{tab:native-results}
\scriptsize
\setlength{\tabcolsep}{3pt}
\renewcommand{\arraystretch}{1.1}
\begin{tabularx}{\columnwidth}{@{}>{\centering\arraybackslash}p{0.32\columnwidth} *{3}{>{\centering\arraybackslash}X}@{}}
\toprule
\textbf{x402 configuration} &
\textbf{Unpaid effect} &
\textbf{Wrong resource} &
\textbf{Paid, no effect} \\
\cmidrule{2-4}
Default & $2/2$ & $1/1$ & $0/2$ \\
Upfront & $0/2$ & $1/1$ & $2/2$ \\
Upfront and binding & $0/2$ & $0/1$ & $2/2$ \\
\midrule
\textbf{ERC-8183 evaluator} &
\textbf{Nonconf. paid} &
\textbf{Conf. paid} &
\textbf{Late paid} \\
\cmidrule{2-4}
Address-controlled & $12/12$ & $8/16$ & $8/12$ \\
Objective validation & $0/12$ & $8/16$ & $0/12$ \\
Provider-callable & $0/12$ & $16/16$ & $0/12$ \\
\bottomrule
\end{tabularx}
\begin{tablenotes}[flushleft]
\scriptsize
\item[-] Wrong resource: resource substitution; Nonconf./Conf.: nonconforming/conforming.  Setup and case definitions: Appendix~\ref{app:experiments}.
\end{tablenotes}
\end{threeparttable}
\end{table}

\keypoint{Payment Order and Service Delivery.}
The x402 cases in Table~\ref{tab:native-results} distinguish payment order, request binding, and service delivery. Two injected settlement races leave a completed request recorded by the service despite failed payment and a suppressed response. Upfront payment prevents unpaid effects but can leave payment without successful service delivery. The tested retry does not restore service. A payer-signed request commitment separately prevents substitution between equally priced resources. Payment order therefore governs whether an effect can remain unpaid, while request binding governs which resource a payment authorizes. The v2 specification distinguishes authorization, upfront, and escrow flows~\cite{x402-spec-v2}. The three policies test authorization, upfront payment, and an extension that binds requests to payments.

\keypoint{Acceptance Correctness and Availability.}
{The ERC-8183 cases separate predicate correctness from invocation access. Before expiry, address-controlled evaluation pays all $12$ nonconforming submissions and eight of $16$ conforming ones (Table~\ref{tab:native-results}). Objective validation prevents all $12$ invalid payments and removes the operator's authority to reject conforming work directly. Opening the same verifier to provider invocation permits all $16$ conforming submissions to be paid. These cases establish a working acceptance path under the tested balances, local transaction inclusion, and pre-expiry scheduling. A complete P9b claim additionally relates the full task cost and worst-case response time to the task's bounds.}

{The native escrow permits completion after expiry while funds remain, and both evaluator extensions require completion before expiry. Terminal rejection or refund closes the payment path. Provider invocation therefore depends on timely inclusion before either terminal event. 
Appendix~\ref{app:native-costs} reports both completion/refund orders and the invocation gas costs.
}

\keypoint{Agent Execution Through Settlement.}
Agent execution introduces a submission interface between correct computation and settlement. We test local Qwen3.5 4B and Gemma 4 E4B together with six hosted models. Each model uses local tools for eight document-hashing tasks. Appendix~\ref{app:agent-execution} specifies the individual and paired inputs. All 64 attempts obtain correct tool results, but only 55 final answers satisfy the strict submission format, as reported in Table~\ref{tab:agent-results}. A normalization rule fixed after the local strict results and before hosted inference extracts one digest object without changing its value. It recovers eight local Qwen3.5 4B submissions and one hosted Qwen3.6 Plus submission. Reusing the same answers across two interfaces, four conditions, and two acceptance policies gives 1,024 native workflows.

After normalization, address-controlled approval pays all 64 cases with altered input evidence and all 64 cases with a replaced output digest. Provider-callable objective verification rejects both substitution groups and pays all 64 conforming outputs during operator silence. Correct computation and a usable submission therefore do not establish task binding or access to acceptance. Appendix~\ref{app:agent-execution} gives the pairing, parsing rules, and full results.

\insightbox{\textbf{Correct verification and usable acceptance are separate requirements.} Task and output checks prevent unrelated evidence from authorizing payment, while a callable path determines whether conforming work can obtain acceptance before the task deadline. Closure requires both the correct predicate and a usable path to its decision.}

\section{Deployment Evidence and Evaluator Decisions}
\label{sec:measurement}


\ifdefined\empiricalAbox\else
  \newsavebox{\empiricalAbox}
  \newsavebox{\empiricalBbox}
  \newsavebox{\empiricalCbox}
  \newsavebox{\empiricalDbox}
  \newlength{\empiricalPanelUnit}
\fi
\begin{figure*}[t]
\centering
\setlength{\fboxrule}{0.6pt}
\setlength{\fboxsep}{1.8pt}
\setlength{\empiricalPanelUnit}{%
  \dimexpr(\textwidth-12pt-6\fboxsep-6\fboxrule)*100/47852\relax}
\captionsetup[subfloat]{position=bottom,justification=centering,
  singlelinecheck=false,labelfont=md,textfont=md,farskip=2pt,nearskip=4pt}
\setcounter{subfigure}{0}
\sbox{\empiricalAbox}{%
  \subfloat[Feedback evidence\label{fig:empirical-feedback}]{%
    \fbox{\includegraphics[width=166\empiricalPanelUnit]{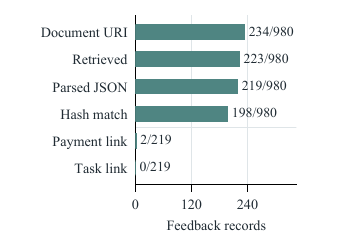}}}}
\setcounter{subfigure}{3}
\sbox{\empiricalDbox}{%
  \subfloat[Recovered settlement decisions\label{fig:empirical-settlement}]{%
    \fbox{\includegraphics[width=166\empiricalPanelUnit]{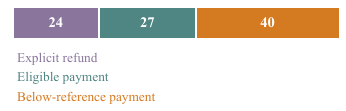}}}}
\setcounter{subfigure}{1}
\sbox{\empiricalBbox}{%
  \subfloat[Reference disagreement\label{fig:empirical-disagreement}]{%
    \fbox{\includegraphics[width=231\empiricalPanelUnit]{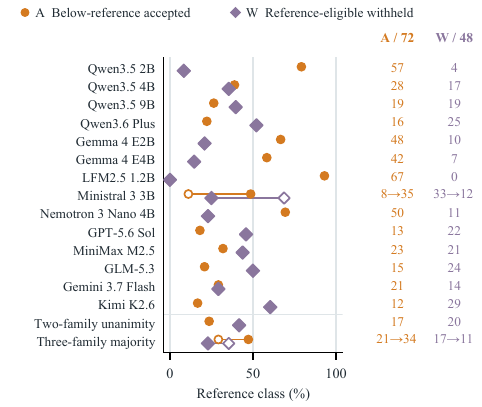}}}}
\setcounter{subfigure}{2}
\sbox{\empiricalCbox}{%
  \subfloat[Unavailable\label{fig:empirical-unavailable}]{%
    \fbox{\includegraphics[width=81.52\empiricalPanelUnit]{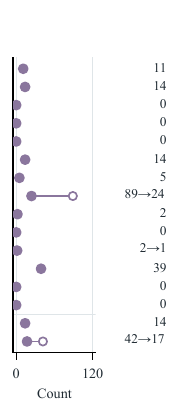}}}}
\noindent\hbox to \textwidth{%
\raisebox{\dp\empiricalDbox}{%
\vbox to \dimexpr\ht\empiricalBbox+\dp\empiricalBbox-\dp\empiricalDbox\relax{%
  \kern2pt
  \hbox{\usebox{\empiricalAbox}}%
  \vfil
  \hbox{\usebox{\empiricalDbox}}}}%
\hfill
\raisebox{\dp\empiricalBbox}{\usebox{\empiricalBbox}}%
\hfill
\raisebox{\dp\empiricalCbox}{\usebox{\empiricalCbox}}%
}
\par\nointerlineskip
\setcounter{subfigure}{3}
\caption{\textbf{Empirical assessment of agentic commerce.}
}
\label{fig:empirical-results}
\end{figure*}

Recorded approval and evidence of task conformance are distinct assessment targets. We examine what public records establish, whether model judgments are both usable and consistent with an independent reference, how committee rules affect those judgments, and how the resulting decisions determine settlement.

\keypoint{Public Records and Task Binding.}
Among the first 10,000 ERC-8004 Ethereum identities, 980 feedback records for 628 identities match native events~\cite{liu2026agentdataset}. We retrieve 223 contents, parse 219, and verify 198 content commitments under the specified byte-decoding rules. The 219 documents cover $12$ agents and come from 178 issuers. None comes from the issuer responsible for 645 events in the full cohort. Four explicitly concern endpoint health checks and two contain structured payment references. The remaining records do not establish a task purpose from these fields alone. No parsed document supplies an identifier in the searched task/job fields. Figure~\ref{fig:empirical-results}(a) therefore measures available linkage evidence. 
Appendix~\ref{app:public_feedback_data} gives the commitment-matching and task-link details.

Two commitment-matching documents link one identity's feedback issuer to confirmed USDC payments to the Base x402 wallet declared in its registration metadata. Neither contains an output commitment, and one incorrectly names Ethereum. The commitment verifies the retrieved content but not the claims within it. Establishing the payment relationship requires the correct chain, while linking payment to completed work further requires task and output bindings.

Thirty-six local cases use the reference ERC-8004 registries~\cite{erc8004implementation}. In eight linkage cases, the registry accepts both the consistent record and seven records with missing or inconsistent external links. An application-level check accepts only the consistent case. The registry enforces its storage and authorization rules, while the application establishes the links to work and payment. Appendix~\ref{app:native-erc8004} describes the cases.

\keypoint{Reference Agreement and Decision Availability.}
Eight local models and six hosted models evaluate existing answers to the same 120 distinct HelpSteer2 requests~\cite{wang2024helpsteer2}. The sample contains 24 answers at each human helpfulness score from zero through four. Scores of at least three define 48 reference-eligible answers, leaving 72 below reference. Models receive the same rubric without human scores, and a usable model score of at least three counts as approval. Normalization extracts one unambiguous score object without changing its value. The rule is developed on local outputs and fixed before hosted inference. Figure~\ref{fig:empirical-results}(b)(c) reports the individual evaluators and two local committees. 
Appendix~\ref{app:evaluator-decisions} gives the sampling rule, generation settings, and repeat-call observations.

Helpfulness at cutoff three supplies an operational reference for this balanced five-stratum sample. The resulting rates characterize that mixture. A contractual P9a assessment additionally requires a task acceptance predicate and adjudication of reference uncertainty, including the score-two/three boundary.

Individual evaluators differ in both below-reference approval and eligible withholding. Qwen3.5 2B, 4B, and 9B approve 57, 28, and 19 of the 72 below-reference answers while withholding 4, 17, and 19 of the 48 reference-eligible answers. Hosted Qwen3.6 Plus approves 16 below-reference answers and withholds 25 eligible answers, with all 120 scores usable. Across these Qwen observations, fewer below-reference approvals accompany more eligible withholdings. The hosted comparison changes model generation and serving configuration along with model size.

Decision availability is distinct from reference agreement. GLM-5.3 has 39 unavailable decisions: 38 exhaust the shared 1,024-token reasoning and output budget, and one has an invalid format. These observations reflect the tested serving configuration. Normalization recovers 65 Ministral scores, including 21 reference-eligible and 27 below-reference approvals. Holding outputs fixed isolates the parsing effect. The rule was developed on local outputs and fixed before hosted calls, with both collections using the same requests.

\keypoint{Committee Composition and Voting Rules.}
The primary unanimity pair uses Qwen3.5 9B and Gemma 4 E4B, requiring usable votes from both and approval by both. A missing vote leaves the decision pending. The pair retains $17$ of Qwen's $19$ below-reference approvals because the other two lack a usable Gemma vote. Among $62$ below-reference responses with usable votes from both models, Gemma opposes none of Qwen's 17 approvals. The drop from $19$ to $17$ therefore reflects decision availability rather than opposing judgments.

The primary strict-majority committee adds Ministral 3 3B. Across the fixed local collection and requests, normalized majority committees of sizes one, three, five, and seven have median below-reference approval counts between 44 and 45 out of 72, with a narrower range for larger committees. Five local models individually approve more than half of the below-reference answers. Committee outcomes can reflect both individual approval rates and dependence between votes. Separating these effects requires comparison with votes that retain the observed marginal approval rates while removing dependence.

Voting rules change this effect. Adding Ministral to the normalized Qwen3.5 9B and Gemma 4 E4B pair reduces below-reference approvals from 17 to 13 under unanimity while increasing eligible withholdings and unavailable decisions. Under strict majority, the same addition increases below-reference approvals from 17 to 34 while reducing both withholding and unavailability. A third usable vote can change both the direction and availability of the committee decision. 
Appendix~\ref{app:paired-diagnostics} reports the sweep over committee membership.

These paired results quantify membership, voting-rule, and availability effects on the 120 observed requests. The committees reuse individual votes and the settlement experiments reuse captured decisions. Proposition~\ref{prop:correlated} below states the common-error and dependence conditions under which a committee soundness bound can instead be derived.

\keypoint{From Evaluator Decisions to Settlement.}
We hold the captured judgments fixed and pass them through native escrow. Each of the $120$ requests is settled under $14$ individual evaluator rules and two local committee rules, yielding $1{,}920$ request-policy pairs per interface and $3{,}840$ settlements across the strict and normalized interfaces. Both interfaces reuse the same $1{,}680$ captured model outputs. Clients commit the observed answers, evaluator transactions carry approval or rejection, and unavailable decisions remain pending until expiry. Human ratings serve only as independent references.

Strict decisions produce $952$ payments, $736$ explicit refunds, and $232$ expiry refunds. Under normalization, the totals become 1,019 payments, 760 explicit refunds, and $141$ expiry refunds. Normalization therefore converts $91$ expiry refunds into $67$ payments and $24$ explicit refunds. The additional payments include $27$ reference-eligible and forty below-reference answers. 
Figure~\ref{fig:empirical-results}(d) shows these settlement changes.

All $3{,}840$ settlements follow the supplied decisions. Explicit rejection produces an immediate refund, while an unavailable decision remains pending until expiry. Making an existing judgment usable can therefore change settlement without changing the recorded judgment. Agreement with the independent reference remains a separate property.

\keypoint{Transcript Evidence and Evaluator Correctness.}
ERC-8183 permits a deliverable reference without requiring content addressing~\cite{crapis2026erc8183}. The same recorded reference and approval can therefore accompany conforming content in one protocol-valid execution and nonconforming content in another. The ledger records the decision and its issuer, but these fields alone cannot distinguish the two cases.

\begin{proposition}[Transcript non-identification]
\label{prop:ledger-nonid}
Let $\mathcal T$ be a settlement transcript and let $Y=1$ denote task conformance, with $Y=0$ otherwise. If the schema admits two protocol-valid worlds with the same approved transcript $\mathcal T$ but different $Y$, the transcript cannot uniquely determine whether the approval was wrongful.
\end{proposition}

A content-addressed proof checked against a declared predicate rules out this construction for that predicate. The evidence retained in the transcript therefore determines which claims about evaluator correctness can be established.

\gapbox{\textbf{Recorded approval does not establish task conformance.} The same approved transcript can accompany conforming and nonconforming work when the recorded evidence does not bind the decision to task conformance. Evaluator correctness therefore requires task-bound conformance evidence beyond the approval record itself.}

\section{Economic Guarantee Analysis}
\label{sec:econ}

Three dependencies govern the economic analysis. Language determines reports before allocation, enforceable penalties determine the cost of purchased votes, and shared errors determine the value of additional evaluators. We state the assumptions and conclusions here. 
Appendix~\ref{app:econ-full} places each model beside its result and proof.

\subsection{Report Formation and Incentives}

For a fixed report policy, write $b_i=v_i+\delta_i(m)$ with $|\delta_i(m)|\leq\delta$. Admissible messages share the task-fact projection $\phi(m)=\phi(m_0)$, and the truthful reference has $\delta_i(m_0)=0$. LLM negotiation experiments motivate report sensitivity~\cite{bianchi2024negotiation}. The bound and which shifts are reachable are modeling assumptions. Truthful reporting remains dominant in the classical auction game, while A1 concerns how a fixed policy forms the submitted bid.

For a second-price auction without reserves, report caps, or withdrawal, Proposition~\ref{prop:ic} gives welfare loss at most $2\delta$. The worst-case construction assumes suitable valuations and reachability of every shift in $[-\delta,\delta]^n$. If a seller can induce a common upward shift of $\delta$, allocation is unchanged and revenue rises by $\delta$ without knowledge of the value profile. This construction can violate ex-post individual rationality when the price increase exceeds the winner's value advantage. Appendix~\ref{app:auction-context} relates these bounded-report consequences to uncertain valuations and auctioneer deviation.

Definition~\ref{def:gev} measures generative extractable value (GEV) as the increase in maximal expected gross value over an honest, fact-only baseline when nested control over ordering, messages, and context becomes available. With finite suprema, the increments are nonnegative. The common-shift example gives total $\mathrm{GEV}\geq\delta$. It bounds the marginal message increment only if the ordering premium is zero. This separates report manipulation from conventional ordering extraction~\cite{daian2020flash,wang2022exploring}.

Corollary~\ref{cor:restore} supplies a sufficient P8 condition. A precommitted deterministic policy $b_i=f_i(v_i,\tau)$ uses a private value fixed independently of the message and canonical attested context $\tau$ identical for fact-equivalent messages. This prevents message-induced shifts, while allowing strategic baseline policies. Truthful policy choice remains a separate incentive claim. This condition requires a validated mapping from fact-equivalent messages to the same context. Separating instructions from data offers relevant design principles~\cite{debenedetti2025camel,chen2025struq}.

\subsection{Evaluator Incentives and Shared Errors}

An unbonded evaluator whose modeled exposure is future participation value $r$ can be compensated for false approval at cost no greater than $r$. For a fresh identity, this cost can remain small even when the escrow value is large. Proposition~\ref{prop:evaluator} instead considers $t$ required approvals, separately committed job bonds $b$, and a uniform per-false-vote probability $q>0$ of enforced net loss $b$ within the job horizon, regardless of quorum formation. Penalties leave the attacking coalition. Honest evaluators reject nonconforming work in this attack class, false votes bring no extra protocol reward or saving, and bribed evaluators maximize expected monetary payoff. For wrongful release $W$ and a purchased-seat attack $A$,
\[
\mathbb E[\operatorname{cost}(A)]\geq tqb\,\Pr[W\mid A].
\]
Thus $b>g/(tq)$ deters direct attacks with gain $g$. The cost bound depends on funded adversary-owned seats and enforceable bribes, including under random seat selection. Subjective tasks may lack evidence for positive $q$. Concurrent exposure and compensation for mistaken penalties must also satisfy P5/P6. A feasible evaluator market must satisfy these deterrence, participation, and budget conditions jointly.

Buying votes and inducing common errors are different channels. Under a declared task/attack distribution with equal marginal error $0<e<1/2$ and pairwise correlation $\rho\in[0,1)$, Proposition~\ref{prop:correlated} obtains
\[
n_{\rm eff}=\frac{n}{1+(n-1)\rho},\qquad
\sigma^2=\frac{e(1-e)}{n_{\rm eff}}.
\]
Counting ties as error, the homogeneous-error specialization of Ladha's correlated-voting bound~\cite{ladha1992}, obtained by Cantelli's inequality, bounds majority error above by $\sigma^2/[\sigma^2+(1/2-e)^2]$. A constructed committee in this class has error at least $\rho e$ for every $n$. This construction limits what size alone can guarantee over the class. Survey-sampling design effects, ensemble diversity, and dependent voting supply the background~\cite{kish1965,breiman2001,kuncheva2003,boland1989,kaniovski2011}. Independent heterogeneous voting is treated separately~\cite{berend1998}. For P9a, $e$ and $\rho$ must hold conditional on nonconforming work. The variance identity and mixture construction are derived in Appendix~\ref{app:correlation-proof}.

Finally, let $C_{\rm dir}$ and $C_{\rm err}$ be infima of expected cost divided by success probability over pure purchased-seat and pure induced-error attacks. Corollary~\ref{cor:combined} gives $C_{\rm dir}\geq tqb$ and
\[
\min\{tqb,C_{\rm err}\}\leq C_{\rm pure}
=\min\{C_{\rm dir},C_{\rm err}\}\leq C_{\rm err}.
\]
Increasing collateral alone cannot certify a pure-channel attack cost above $C_{\rm err}$. Mixed attacks require a joint model. Report derivation, misconduct evidence, and the relevant error distribution must therefore accompany any economic guarantee.

\section{A Gap-Derived Research Agenda}
\label{sec:agenda}


The identified failures suggest three missing guarantees and three enabling interfaces for end-to-end guarantee closure.

\subsection{Holes and Missing Guarantees}

\keypoint{Sound capability attestation (P2):}
Capability evidence must bind the serving identity and full agent configuration to performance on a declared task class. Existing attested evaluation designs provide part of this chain~\cite{schnabl2025attestable}. Remaining gaps include statistical confidence, tools and memory beyond model weights, and validity under changing workloads.

\keypoint{Persuasion-robust market interfaces (P8):}
Corollary~\ref{cor:restore} fixes report derivation, canonical attested context, and the independently formed value. A usable interface must retain decision-relevant information while excluding irrelevant persuasive content, with its persuasion bound assessed for the deployed model and interface.

\keypoint{Cross-domain conditional atomicity (P12):}
Conditional-payment and integrated execution-payment designs provide starting points~\cite{campanelli2017zkcp,li2026a402}, subject to known attacks on earlier ZK contingent-payment repairs~\cite{fuchsbauer2019wi}. Across trust domains, the remaining obligation is to preserve the specified acceptance predicate, authentication, finality, expiry, and refunds. Existing coupling for a measured program does not alone establish subjective task conformance.

\subsection{Interfaces and Enabling Mechanisms}

\keypoint{Receipt-consuming settlement:}
When adequate verification is available, requiring payment to check it repairs the payment-first break for objective single-domain tasks. The check must validate the receipt's task, output, acceptance condition, and trust model. The upstream mechanism provides soundness, the acceptance service provides the receipt path, and settlement consumes the evidence.

\keypoint{Machine-checkable evidence (P16):}
A common evidence object should link the task, deliverable digest, procedure, receipt, and evaluator decision, with authenticated provenance and deterministic replay. ERC-8183 extension points can carry but do not define this object~\cite{crapis2026erc8183}. For subjective work, it must also record assessment evidence. P9a separately asks whether that evidence supports conformance.

\keypoint{Economically sound evaluator markets (P7$\rightarrow$P9a):}
Propositions~\ref{prop:evaluator},~\ref{prop:correlated} and Corollary~\ref{cor:combined} show why evaluator design must address capture cost, enforceable penalties, and shared errors. Job-specific misconduct evidence is required before a penalty probability can justify the incentive bound; increasing stake alone does not establish that probability.
P9b also requires a feasible acceptance path. Penalizing approval but not rejection may favor conservative refusal. Compensation, bounded-cost appeals, or alternative verification can support timely acceptance. Provider-callable verification must occur before terminal refund.

\section{Related Work}
\label{sec:related}

Zhang~\etal analyze intent binding, payment--service separation, and lifecycle consistency, identifying multi-round negotiation of payable terms as an open direction~\cite{zhang2026a2apay}. Mao~\etal examine evaluator refusal, collusion, shared model failures, and cross-layer propagation~\cite{mao2026sok}. Analyses of x402 examine payment-security conditions, discovery attacks, settlement inconsistencies, and implementation failures~\cite{li2026x402attacks,ling2026freeriding,wang2026http402}, with conditional security arguments for payment workflows and composed defenses~\cite{li2026x402attacks,ling2026freeriding}. Our analysis relates these findings to task-specific obligations, evidence consumption, and repair locations.

Xiong~\etal study task/payment linkage in public feedback~\cite{xiong2026erc8004}. We extend their analysis by checking feedback-content commitments, stratifying retrieved evidence by issuer, and testing application-level linkage checks. Table~\ref{tab:related} summarizes literature overlaps.

Assume-guarantee reasoning supports compositional verification under component assumptions~\cite{henzinger2002assume}, while universally composable security provides a simulation-based framework for security-preserving composition~\cite{canetti2001uc}. Allegrini~\etal formalize properties of agent orchestration and task lifecycles~\cite{allegrini2025formal}. Guarantee closure specifies a task-relative criterion for establishing required guarantees and preserving them at dependent decisions, with explicit evidence bindings, compatible assumptions, validity conditions, and consumer checks.

Classical multi-agent research studies negotiation, allocation, false-name resistance, faithful implementation, and computational trust~\cite{smith1980contractnet,jennings2001negotiation,baarslag2016opponent,yokoo2001aij,shneidman2004faithful,huynh2006fire,teacy2006travos}. Mechanism design provides foundations for analyzing incentive compatibility, individual rationality, and budget balance~\cite{myerson1981,myerson1983,nisan2007}. We examine how persuasion-sensitive reports, enforceable evaluator penalties, and dependent errors constrain task-level guarantees, and derive extraction and corruption-cost bounds under explicit assumptions.

\section{Conclusion}
\label{sec:conclusion}


We systematize decentralized agent economic infrastructure and introduce guarantee closure across the task lifecycle. Our analysis shows that correct protocol execution can still pay for nonconforming work or leave conforming work unpaid. Targeted repairs address specific dependencies but leave further requirements for closure. We also show that recorded approval alone does not establish task conformance, while economic guarantees depend on report formation, enforceable penalties, and evaluator error dependence. Secure agentic commerce requires guarantees to be established for the task and constrain dependent downstream decisions.



\label{paper:body-end}


\section*{Open Science}
We provide anonymized source code, prompts, fixed study inputs, model configurations, protocol sources, and selected result summaries. These materials support rerunning the controlled experiments and reconstructing the evaluator sample. 

\section*{LLM Usage Considerations}
LLMs were used for editorial purposes in this manuscript, and all outputs were
inspected by the authors to ensure accuracy and originality. LLMs were also used
to simulate reviewer feedback on drafts, and all resulting changes were made and
verified by the authors.

\bibliographystyle{IEEEtran}
\bibliography{refs}
\normalcolor

\appendices


\section{Formal Property Definitions}
\label{app:formal}

The property definitions fix what a guarantee establishes, the evidence it requires, and the assumptions and bounds under which it holds. The security parameter is $\lambda$, $\mathrm{negl}(\lambda)$ denotes a negligible function. PPT means probabilistic polynomial time.

\subsection{S1 Discovery}

\keypoint{P2. Capability Soundness}
\emph{Formalization.}
For identity $\mathsf{id}$ advertising capability $c$ with attestation $\pi$, P2 binds statistical performance to the serving configuration. Let $h$ identify the model, code, tools, and configuration, $I$ the validity interval, and $\mathcal D_c^v$ the versioned reference distribution. Let $\mu_h$ be the true expected task score under the fixed serving and sampling law during $I$, integrating service randomness. The bundle, claimed level $\theta$, distribution, and evaluation procedure are committed before fresh evaluation data are sampled. For tolerance $\eta\geq0$ and per-evaluation error budget $\beta\in[0,1]$, define
$V_\pi=\mathsf{Verify}(\mathsf{id},h,I,c,\theta,\mathcal D_c^v,\pi)$.

For every allowed adversary respecting this evaluation protocol, soundness requires
$\Pr[V_\pi=1\ \land\ \mu_h<\theta-\eta]
\leq\beta+\mathrm{negl}(\lambda)$.
Probability covers fresh evaluation samples and protocol randomness. The serving configuration and reference law must remain stable throughout $I$.

The service must also connect requests authenticated to $\mathsf{id}$ to the measured configuration $h$ throughout $I$, for example by remote measurement and a serving key bound to that configuration. Non-transferability means that no PPT adversary can rebind an existing attestation to a different identity or serving configuration except with negligible probability.

\subsection{S2 Negotiation}

\keypoint{P7. Collusion Resistance}
\emph{Formalization.}
Let $h$ denote the mechanism's outcome rule, let $u_i$ denote the utility of participant $i$ under the resulting outcome, and let $\sigma^*$ denote the intended strategy profile. For a declared coalition-size bound $\tau\geq1$ and gain bound $B\geq0$, the mechanism is \emph{$(\tau,B)$-collusion-resistant} if the inequality below holds for every coalition $C$ with $|C|\leq\tau$ and every coalition strategy $\sigma_C$ in the declared admissible strategy class.
$\mathbb{E}\!\left[\sum_{i\in C} u_i\!\left(h(\sigma_C,\sigma^*_{-C})\right)\right] - \mathbb{E}\!\left[\sum_{i\in C} u_i\!\left(h(\sigma^*)\right)\right] \leq B.$
The bound $B$ applies only to the stated coalition size, utility model, and strategy class. Corollary~\ref{cor:combined} adds a separate induced-error channel, so a P7 claim must identify which strategies and assumptions it covers.

\keypoint{P8. Persuasion Robustness}
\emph{Formalization.}
Let $\phi(m)$ denote the task-relevant fact projection of message $m$, and let $\beta_i(v_i,m)$ denote participant $i$'s realized report or decision given private type $v_i$ and message $m$. A market interface is \emph{persuasion robust} if
$\beta_i(v_i,m)=\beta_i(v_i,m')$
whenever
$\phi(m)=\phi(m').$
The comparison holds task-relevant facts fixed. For a general policy, P8 requires $\Delta_i=0$ relative to a fact-equivalent reference message, while baseline strategic shading may remain. In the truthful-reference subclass of Eq.~\eqref{eq:persuadable}, this requires $\delta_i(m)=0$ for every admissible message. Randomized policies require equal conditional output distributions. Corollary~\ref{cor:restore} obtains this invariance by deriving reports from canonical attested facts.

\subsection{S4 Verification}

\keypoint{P9a. Receipt Soundness}
\emph{Formalization.}
Let $\psi_T(y)=1$ mean that output $y$ satisfies task $T$. For every PPT adversary producing $(y,\pi)$, cryptographic soundness requires
$\Pr[\mathsf{Verify}(T,y,\pi)=1\land\psi_T(y)=0] \leq\mathrm{negl}(\lambda).$
The probability covers setup, adversarial, and verifier randomness \cite{groth2016}. {A statistical acceptance claim instead specifies an error bound and the task and adversary distribution over which it holds.} An economic relaxation instead states a minimum attack cost for a specified strategy class. Every P9a claim must therefore identify its strategy class, trust assumptions, and evidence process.

\keypoint{P9b. Receipt Completeness}
{For a conforming output, P9b requires a permitted honest-provider strategy that obtains accepted, task-bound evidence within the task's time and total cost bounds against every adversary admitted by $\Gamma$. The strategy, available invocation and retry actions, response assumptions, and terminal conditions form part of the claim. For deadline $\tau$, budget $B$, and availability error $\epsilon_b$, a strategy $\sigma_{\rm P}$ must satisfy
$\Pr_{\sigma_{\rm P},\mathcal A,\Gamma} [\text{matching acceptance by }\tau\ \land\ C_{\rm P}\leq B] \geq 1-\epsilon_b$
for every admitted $\mathcal A$, conditional on a conforming submission. Probability covers scheduling and mechanism randomness under the declared model. A deterministic claim sets $\epsilon_b=0$. An unsuccessful execution of one policy supplies evidence about that policy. Refuting P9b requires an admitted counterstrategy for every permitted provider strategy. A hole identifies the absent acceptance or recovery path. Client refunds, delivery of an existing receipt, and timely provider payment have separate roles under P12 and P13.}

\subsection{S5 Settlement}

\begin{table*}[t]
\centering
\begin{threeparttable}
\caption{Verification conditions for P9a and P9b.}
\label{tab:verif-families}
\scriptsize
\setlength{\tabcolsep}{4pt}
\renewcommand{\arraystretch}{1.15}
\begin{tabularx}{\textwidth}{@{}c *{3}{>{\raggedright\arraybackslash}X}@{}}
\toprule
\textbf{Family} &
\multicolumn{1}{c}{\textbf{Acceptance condition}} &
\multicolumn{1}{c}{\textbf{Soundness assumptions (P9a)}} &
\multicolumn{1}{c}{\textbf{Availability constraints (P9b)}} \\
\midrule

ZK receipt &
Proof of committed I/O relation &
Proof-system soundness and commitment binding &
Proof/verifi. meet task time/cost bounds~\cite{sun2024zkllm,qu2025zkgpt} \\
\addlinespace[2pt]

TEE &
Attested output from committed, measured code that checks task acceptance &
Hardware trust, measurement integrity, and attestation-key security &
Inference and attestation overhead; remediation may delay or block receipts~\cite{vanschaik2024sgxfail} \\
\addlinespace[2pt]

Optimistic receipt &
Assertion accepted after the window with no successful challenge &
An honest, funded, live challenger able to detect the modeled invalid claim &
Assertion, challenge, and dispute must finish within the task horizon~\cite{lee2026challenge} \\
\addlinespace[2pt]

Evaluator verdict &
Designated evaluator or committee approves &
Evaluator accuracy and independence, or an economic bound for the stated attack class &
Bounded evaluation and response; silence rules determine recovery~\cite{crapis2026erc8183,okx-app,okx-app-whitepaper} \\

\bottomrule
\end{tabularx}
\begin{tablenotes}[flushleft]
\scriptsize
\item[] P9a: receipt soundness; P9b: receipt completeness. Acceptance must match the task predicate; an honest provider must obtain it within the task's time and total cost bounds.
\end{tablenotes}
\end{threeparttable}
\end{table*}

\keypoint{P12. Cross-Domain Conditional Atomicity}
When verification and payment occur in different security domains, P12 keeps payment conditional on acceptance of the same task. Safety permits payment only after matching finalized acceptance. Under the declared delivery and finality bounds, eligible acceptance must lead to payment within the task horizon. Expiry and refund rules exclude duplicate release and conflicting terminal outcomes \cite{zamyatin2021sok}.

\keypoint{P13. Timeliness and Progress}
{P13 requires the selected workflow policy to reach the task's specified completion outcome within its horizon under the declared delivery, finality, and service assumptions. In the controlled tasks, completion means payment to a conforming provider by time 16. Acceptance acquisition, receipt delivery, and payment are measured separately. A conforming execution that satisfies the progress premises and misses that deadline refutes the selected policy's P13 guarantee. Tasks with incremental settlement additionally require payment for verified partial results. A terminal client refund is a distinct outcome.}

\subsection{Aggregation of Closure Bounds}
\label{app:closure-bounds}
Closure retains the bound attached to each active obligation. To combine probability guarantees, let $F_j$ be bad events in the same execution with justified bounds $\Pr(F_j)\leq\epsilon_j$. Adaptive calls require the corresponding bounds conditional on admissible preceding histories. For task error budget $\epsilon_{\rm task}$, the union bound gives
$\Pr\!\left[\bigcup_j F_j\right]\leq\sum_j\epsilon_j\leq\epsilon_{\rm task}.$
The bound permits dependent events. Time and economic guarantees are checked against thresholds in their respective units.

\section{Corpus Construction and Scoring}
\label{app:corpus}

Each rating retains its assessment unit, evidence source, and assumptions. The corpus rules fix these choices.

\subsection{Corpus and Scoring Rules}

We derived the properties by identifying distinct security or economic losses across the six lifecycle stages. The sources cover decentralized finance (DeFi), MEV, multi-agent systems, and agent security \cite{werner2022sok,daian2020flash,mao2026sok}, together with conditions A1 and A2 in \S\ref{sec:threatmodel}.


We use purposive selection from the comparative studies in Appendix~\ref{app:related_works}, assessing candidates available by 31 August 2026 against the lifecycle criterion using primary, version-specific sources and archived snapshots. Public research prototypes are eligible, including A402 alongside standards and platforms. Wallet-only ERC-4337 infrastructure \cite{erc4337} and closed platforms without an exposed lifecycle stage are outside scope, while brokered OKX is an explicit comparator.

For P9a, approval or task payment is assessed on whether its evidence establishes conformance. For P9b, we identify the provider's acceptance route, who controls it, and the time and resource conditions. A submission interface alone supplies no availability guarantee when an exclusive evaluator can withhold its decision. This rule gives both ACP and ERC-8183 unsupported P9b. AP2, x402, and MPP provide payment authorization rather than an acceptance service, so P9b is inapplicable to those component interfaces. Service platforms such as Fetch.ai and SingularityNET retain P9b as a task requirement. Restricted P9b requires a substantive progress mechanism under additional assumptions, such as an attested exchange, a specified challenge period, or provider delivery within a cooperative service. Table~\ref{tab:acceptance-evidence} applies the same rule to the core acceptance and evidence claims.

\subsection{Evidence Basis for the Mechanism Matrix}
\label{app:rationale}

\begin{table*}[t]
\color{black}
\centering

\begin{threeparttable}
\caption{EVIDENCE FOR ACCEPTANCE AND DECISION-RECORD RATINGS}
\label{tab:acceptance-evidence}
\scriptsize
\setlength{\tabcolsep}{4pt}
\renewcommand{\arraystretch}{1.10}
\begin{tabularx}{\textwidth}{@{}rcccX@{}}
\toprule
\multicolumn{1}{c}{\textbf{Assessed unit}} & \textbf{P9a} & \textbf{P9b} & \textbf{P16} & \multicolumn{1}{c}{\textbf{Decisive evidence and retained condition}} \\
\midrule
ERC-8004 registry & R & \outc{} & R & Validator responses and attributed feedback support provenance. Task conformance depends on the validator and linked evidence. The registry does not offer an acceptance service~\cite{erc8004}. \\
ACP task protocol & R & U & R & Agreement, evaluation, and payment records bind a task history. Quality remains evaluator-dependent, and the documented evaluation route provides no bounded alternative when the evaluator withholds approval~\cite{virtuals-acp}. \\
ERC-8183 job escrow & R & U & R & Submission and evaluator-controlled terminal transitions attribute the verdict. The evaluator supplies quality evidence, while expiry returns client funds. Provider access to acceptance remains unestablished~\cite{crapis2026erc8183}. \\
AP2 payment mandates & U & \outc{} & R & Signed mandates establish payment authority and part of the decision history. The assessed interface supplies neither a conformance test nor an acceptance service~\cite{ap2}. \\
x402 payment and receipt & U & \outc{} & R & Signed offers and transfer or delivery statements bind payment terms. Task-conformance evidence and an acceptance service are separate workflow requirements~\cite{x402-spec-v2,x402-offer-receipt}. \\
A402 attested exchange & R & R & R & Attested computation and request-bound exchange support a receipt under hardware trust. Exchange progress retains eventual-delivery and resource assumptions, while task-specific receipt deadlines and total cost bounds require separate support~\cite{li2026a402}. \\
MPP payment interface & U & \outc{} & U & Payment authorization and metering govern access and amounts. The assessed payment interface supplies no task acceptance or complete decision reconstruction~\cite{mpp-docs,httpauth-payment}. \\
OKX APP proposal & R & R & R & Escrowed exchange, dispute-window release, and signed settlement records provide conditional support under the proposed broker and dispute rules~\cite{okx-app-whitepaper}. \\
Bittensor scoring & R & R & U & Stake-weighted evaluation and periodic reward allocation retain honest-validator and participation assumptions. The record of weights does not reconstruct each task judgment~\cite{bittensor-docs,bittensor-validating}. \\
Olas service marketplace & R & R & R & Bonded service and provider delivery support conditional completion, while request and delivery records preserve part of the task history. Quality and progress retain service-participation assumptions~\cite{autonolas,olas-mech}. \\
Fetch.ai registration & U & U & U & Registration and authenticated messages identify participants. The assessed interfaces establish no deliverable acceptance service or complete evaluation history~\cite{fetchai-almanac,uagents}. \\
SingularityNET registry & U & U & U & Service discovery and payment channels establish transfer infrastructure. Acceptance quality, bounded receipt acquisition, and decision reconstruction remain application requirements~\cite{singularitynet2019}. \\
\bottomrule
\end{tabularx}
\begin{tablenotes}[flushleft]
\scriptsize
\item[] R: restricted support; U: unsupported; \outc{}: inapplicable component role. Restricted ratings retain the cited subclaim or additional trust assumptions.
\end{tablenotes}
\end{threeparttable}
\end{table*}

\keypoint{Cross-Domain Settlement.}
HTLCs condition payment on preimage revelation and require cross-domain liveness within timelocks~\cite{xu2021htlc}. Validator or proof-based bridges authenticate source state under validator-set or consensus/proof assumptions~\cite{zamyatin2021sok}. Optimistic bridges require an honest, funded watcher and a sufficient challenge window. Intent settlement validates a destination fill before repayment, retaining solver-liquidity, validation, and challenge assumptions~\cite{across-intent-lifecycle,erc7683}. For P12, each payment condition must bind finalized acceptance of the same task and result; delivery, challenge, expiry, and refund timing must fit the task horizon.

\keypoint{Registry and Task Protocols.}
ERC-8004 authenticates identities, feedback authors, validators, and validation responses, providing the provenance and reuse needed for P3. Its P3 support remains restricted because proof that feedback refers to a completed interaction is optional. The substantive correctness of the rating is a separate P9a question and requires task evidence. The owner-initiated validation interface similarly provides only restricted P9a and P16.

\keypoint{Classical Mechanisms.}
False-name-proof auctions establish full P1 and P4 under their stated restrictions on both true and admissible reported valuation types. Their P7 result is narrower. It uses $\tau=1$ and $B=0$ for one economic participant operating multiple identities \cite{yokoo2001aij,yokoo2004geb}. Faithful distributed implementation similarly establishes full P4 through ex post Nash incentives for reporting, communication, and computation. Its P7 result also uses $\tau=1$ and $B=0$, with the other participants following the prescribed strategy. These results support restricted P7 because collusion among multiple economic participants requires a separate coalition bound. Computational trust supports restricted P1 for the value of each controlled alias under the stated trust-root and path assumptions, without bounding their aggregate influence~\cite{cheng2005sybilproof}.

\begin{table*}[t]
\centering
\begin{threeparttable}
\caption{Analysis coverage in closely related work.}
\label{tab:related}
\setlength{\tabcolsep}{2pt}
\renewcommand{\arraystretch}{1.3}
\newcommand{\relatedrot}[1]{%
  \makebox[\linewidth][l]{\hspace{0.15cm}%
    \rotatebox{30}{\makebox[3cm][l]{\scriptsize\strut #1}}}%
}
\setlength{\bigstrutjot}{\dimexpr2.75cm-\baselineskip+\aboverulesep+\belowrulesep+\cmidrulewidth\relax}
\newcommand{\relatedhead}[1]{\multirow{2}[1]{*}[10pt]{#1}}
\begin{tabular*}{\textwidth}{@{\extracolsep{\fill}}>{\centering\arraybackslash}p{0.8cm} >{\raggedleft\arraybackslash}p{3.9cm} >{\centering\arraybackslash}p{0.7cm} *{14}{C{0.55cm}}@{\hspace{1.1cm}}}
\toprule
\multicolumn{1}{c}{\relatedhead{\rot{\textbf{Group}}}}
& \multicolumn{1}{c}{\relatedhead{\textbf{Work}}}
& \multicolumn{1}{c}{\relatedhead{\textbf{Type}}}
& \relatedrot{S1 Discovery}
& \relatedrot{S2 Negotiation/terms}
& \relatedrot{S3 Execution}
& \relatedrot{S4 Verification}
& \relatedrot{S5 Settlement}
& \relatedrot{S6 Dispute/accountability}
& \relatedrot{Formalized properties}
& \relatedrot{Requirement derivation}
& \relatedrot{System--property matrix}
& \relatedrot{Cross-stage dependencies}
& \relatedrot{End-to-end preservation}
& \relatedrot{Strategic incentives}
& \relatedrot{LLM robustness}
& \relatedrot{Empirical validation} \\
\cmidrule(lr){4-9}\cmidrule(lr){10-12}\cmidrule(lr){13-14}\cmidrule(lr){15-16}\cmidrule(lr){17-17}
& & & \multicolumn{6}{c}{\textbf{Task lifecycle}}
& \multicolumn{3}{c}{\textbf{Property analysis}}
& \multicolumn{2}{c}{\textbf{Composition}}
& \multicolumn{2}{c}{\textbf{Agent/econ.}}
& \multicolumn{1}{c@{\hspace{1.1cm}}}{\makebox[0.55cm][c]{\textbf{Evidence}}} \\
\midrule
\multirow{3}{*}{\rotatebox[origin=c]{90}{\shortstack{\textit{Threats}\\\textit{(3)}}}}
& Mao \textit{et al.}~\cite{mao2026sok} & SoK
   & \Cpart & \Cpart & \Cpart & \Cpart & \Cfull & \Cpart & \Cnone & \Cfull & \Cfull & \Cfull & \Cpart & \Cfull & \Cfull & \Cnone \\
& Shi \textit{et al.}~\cite{shi2025trustauthz} & SoK
   & \Cpart & \Cnone & \Cfull & \Cnone & \Cnone & \Cnone & \Cfull & \Cfull & \Cnone & \Cfull & \Cpart & \Cnone & \Cfull & \Cnone \\
& Al Masoud \textit{et al.}~\cite{almasoud2026judge} & SoK
   & \Cnone & \Cnone & \Cnone & \Cfull & \Cnone & \Cnone & \Cnone & \Cpart & \Cpart & \Cpart & \Cnone & \Cnone & \Cfull & \Cnone \\
\midrule
\multirow{8}{*}{\rotatebox[origin=c]{90}{\textit{Agent economy (8)}}}
& Zhang \textit{et al.}~\cite{zhang2026a2apay} & SoK
   & \Cfull & \Cpart & \Cpart & \Cpart & \Cfull & \Cpart & \Cnone & \Cpart & \Cfull & \Cfull & \Cpart & \Cpart & \Cpart & \Cnone \\
& Li \textit{et al.}~\cite{li2026x402attacks} & Ana
   & \Cfull & \Cpart & \Cpart & \Cnone & \Cfull & \Cpart & \Cfull & \Cfull & \Cfull & \Cfull & \Cfull & \Cpart & \Cfull & \Cfull \\
& Ling \textit{et al.}~\cite{ling2026freeriding} & Ana
   & \Cnone & \Cpart & \Cfull & \Cnone & \Cfull & \Cnone & \Cfull & \Cfull & \Cpart & \Cfull & \Cfull & \Cfull & \Cpart & \Cfull \\
& Wang \textit{et al.}~\cite{wang2026http402} & Msr
   & \Cnone & \Cnone & \Cpart & \Cnone & \Cfull & \Cnone & \Cfull & \Cfull & \Cfull & \Cfull & \Cpart & \Cpart & \Cnone & \Cfull \\
& Xiong \textit{et al.}~\cite{xiong2026erc8004} & Msr
   & \Cfull & \Cnone & \Cnone & \Cpart & \Cpart & \Cnone & \Cfull & \Cfull & \Cnone & \Cpart & \Cnone & \Cfull & \Cnone & \Cfull \\
& Gong~\cite{gong2026a2afinance} & Srv
   & \Cfull & \Cpart & \Cpart & \Cpart & \Cfull & \Cpart & \Cnone & \Cfull & \Cnone & \Cfull & \Cpart & \Cpart & \Cpart & \Cnone \\
& Hu and Rong~\cite{hu2025trust} & Srv
   & \Cfull & \Cnone & \Cpart & \Cfull & \Cpart & \Cpart & \Cnone & \Cfull & \Cpart & \Cpart & \Cpart & \Cfull & \Cfull & \Cnone \\
& Zhu \textit{et al.}~\cite{zhu2026agentnet} & Srv
   & \Cfull & \Cpart & \Cfull & \Cpart & \Cfull & \Cfull & \Cnone & \Cfull & \Cpart & \Cfull & \Cpart & \Cpart & \Cfull & \Cnone \\
\midrule
\multirow{3}{*}{\rotatebox[origin=c]{90}{\shortstack{\textit{DeAI /}\\\textit{infra. (3)}}}}
& Lui \textit{et al.}~\cite{lui2026deai}$^{\dagger}$ & SoK
   & \Cpart & \Cnone & \Cfull & \Cfull & \Cpart & \Cpart & \Cpart & \Cfull & \Cfull & \Cpart & \Cnone & \Cfull & \Cpart & \Cfull \\
& Alqithami~\cite{alqithami2026chains} & Srv
   & \Cpart & \Cpart & \Cfull & \Cpart & \Cfull & \Cpart & \Cpart & \Cfull & \Cfull & \Cfull & \Cpart & \Cfull & \Cfull & \Cnone \\
& Touloupou and Kapassa~\cite{touloupou2026review} & Srv
   & \Cpart & \Cfull & \Cfull & \Cfull & \Cfull & \Cfull & \Cnone & \Cpart & \Cfull & \Cfull & \Cpart & \Cfull & \Cpart & \Cnone \\
\midrule
\multirow{2}{*}{\rotatebox[origin=c]{90}{\shortstack{\textit{Formal}\\\textit{gaps}}}}
& Allegrini \textit{et al.}~\cite{allegrini2025formal} & Ana
   & \Cpart & \Cnone & \Cfull & \Cnone & \Cnone & \Cnone & \Cfull & \Cfull & \Cnone & \Cfull & \Cpart & \Cnone & \Cpart & \Cnone \\
& Otsuka \textit{et al.}~\cite{otsuka2026identity} & Srv
   & \Cfull & \Cnone & \Cpart & \Cnone & \Cnone & \Cpart & \Cpart & \Cfull & \Cfull & \Cfull & \Cpart & \Cpart & \Cpart & \Cnone \\
\midrule
& \textbf{This work}$^{\ddagger}$ & SoK
   & \Cfull & \Cfull & \Cfull & \Cfull & \Cfull & \Cfull
   & \Cfull & \Cfull & \Cfull & \Cfull & \Cfull & \Cfull & \Cfull & \Cfull \\
\bottomrule
\end{tabular*}
\begin{tablenotes}[flushleft]
\scriptsize
\item[-] \textbf{Symbols.} \Cfull: primary analysis; \Cpart: limited/secondary analysis; \Cnone: no substantive coverage. Qualitative analysis may be primary.
\item[-] \textbf{Scope.} Scores measure scope-relative coverage, not security, proof quality, or rank. S6 requires post-task accountability beyond identity authentication. Criteria: Appendix~\ref{app:related_works}.
\item[-] \textbf{Types.} Ana: analysis; Msr: measurement; Srv: survey; DeAI: decentralized AI; infra.: agent infrastructure; formal gaps: property formalization and standards-gap analyses.
\item[$\dagger$] Organized around the AI lifecycle.
\item[$\ddagger$] LLM coverage concerns A1/A2; all sources' preservation claims remain relative to declared tasks, properties, and assumptions.
\end{tablenotes}
\end{threeparttable}
\end{table*}

\begingroup

\let\FloatBarrier\relax
\subsection{Literature Crosswalk}
\label{app:related_works}

Table~\ref{tab:related} records analysis coverage, not achieved security or an overall ranking. Full coverage means explicit and substantive primary analysis; partial coverage means limited or secondary treatment; empty coverage means no substantive analysis. A dedicated survey chapter can count as primary coverage, and qualitative analysis or a bounded scope alone does not justify a downgrade. We apply the same rule to this work. Background mentions and undeveloped future work do not count. Concrete requirements or arguments in discussion or research-directions sections can count as analysis, but proposed evaluations do not count as empirical validation.

\section{Controlled Evaluation of Guarantee Closure}
\label{app:closure-eval}

We separate three questions in this section: whether acceptance is established, whether payment consumes it, and whether a conforming provider can still obtain it.

\keypoint{Reference contract and required dependencies.}
\label{app:reference-contract}
{The objective single-domain reference pays 10 units to sort a committed array. Its acceptance-service contract permits direct provider invocation, authenticates the agreed task and output, and returns matching acceptance at time 2 for a conforming submission at time 1. The reference assumes ideal commitment and authentication, a sound deterministic verifier, and this response bound. Authenticated delivery takes at most ten steps, and the next settlement step releases the prefunded payment by time 13, within the deadline of 16. The provider's total cost $c_{\rm P}\leq10$ includes execution, verification, fees, and delay losses. The requester's value $v_{\rm R}\geq10$ is net of its costs, and both outside options are zero.}

The required dependencies are
$\mathcal D_T={(P5,2,5),(P6,2,5),(P9a,4,5),(P9b,4,5),(P13,5,6)}$.
P5 and P6 require settlement to preserve the fixed payment, cost assumptions, and funding cap. P9a requires payment to consume task-bound acceptance. P9b requires a usable acceptance path within the time and fee bounds. P13 requires the deadline to hold through terminal recording.

\subsection{Reference-Task Obligations}
\label{app:closure-ledger}
Table~\ref{tab:closure-ledger} states the reference\textquotesingle s proof obligations. Each configuration counts an active property once, including its required preservation checks. Subjective configurations add P2 at S1 with dependency $(P2,1,2)$ and P7 at S4 with $(P7,4,5)$. Cross-domain configurations add P12 within S5.

\subsection{Tasks and Schedule}
\label{app:simulation-schedule}

The objective task sorts an integer array while preserving its length and element multiplicities. The default input is $(3,1,1,2)$, and the nonconforming output adds an element. An independent observer checks order and multiplicity. For subjective tasks, each deliverable receives a fixed conformance label and a separate evaluator verdict. Varying their agreement isolates the effect of an inaccurate verdict on settlement. Model judgment quality is evaluated separately using the HelpSteer2 answers in \S\ref{sec:measurement}.

Time is measured in logical steps, with explicit cross-domain delivery and source finality. Escrow receives 10 units at time 0, and the provider submits at time 1. At time 2, the evaluator returns a verdict or remains silent. Objective approval must pass the fixed verifier. With fallback enabled, the provider can invoke the verifier after silence or a contestable negative verdict in the same step. Acceptance is finalized at the source at time 2 under the finality assumption, while the original negative record is retained. Subjective acceptance follows the assigned verdict. Both architectures share this acceptance logic. Under the authentication assumption, only the fixed evaluator can issue subjective acceptance, so its silence blocks every provider strategy within the fixed task.

A complete source receipt reaches settlement at $2+d$, and payment is attempted at $3+d$. The delay $d$ includes transport and scheduling in either domain placement. Its default is zero and its admitted bound is 10. Source acceptance and destination availability are due by time 12, payment is due by time 16, and receipt validity includes time 16. Payment checks therefore validate the receipt or task-bound authorization at the actual payment attempt rather than only when the evidence arrives.

\subsection{Perturbations and Repair Policies}
\label{app:repair-design}

The fifteen experimental conditions include honest execution and fourteen perturbations covering a nonconforming output, false approval, rejection of conforming work, evaluator silence, substitution of the task, output, or acceptance predicate, invalid authentication, expired acceptance, missing finality, duplicate settlement, delayed delivery, permanent partition, and altered replay evidence. Each condition is evaluated under seven policies comprising the baseline, receipt checking at payment (C), source-side authorization (A), provider-callable verification (V), decision recording (E), C+A, and all repairs. Policy comparisons are matched within each condition and configuration. Comparisons among the baseline, C, A, and C+A disable V and E so that receipt consumption is isolated from acceptance recovery and replay evidence.

C validates the complete receipt at payment. It authenticates the receipt and checks its provider, task, output, acceptance predicate, producer, finality, and current validity against the payment obligation. Objective acceptance uses the fixed verifier, while subjective acceptance uses the assigned verdict.

The channel delivers the complete receipt from the verification domain to the payment domain at $2+d$, where $d$ is the delivery delay defined in \S\ref{app:simulation-schedule}. The baseline adapter checks only whether funds were delivered. A instead validates the receipt in the verification domain before issuing a task-bound payment authorization. It uses the same receipt checks as C and also requires source finality. The resulting authorization is tied to the verified task and payment and expires with the receipt at time 16.

{V invokes the fixed objective verifier after silence or a contestable negative verdict and before terminal transfer. An initial approval with a substituted binding, invalid authentication, or expired receipt bypasses this invocation rule. C and A subsequently reject that unusable receipt, leaving conforming work unpaid. Recovery for these cases requires invocation after any failed acceptance check, followed by matching replacement evidence before the deadline. The required response and total cost bounds must cover this recovery. E retains decision inputs, procedure versions, and authenticated records for replay.}

\subsection{Measurements and Eligible Cases}
\label{app:simulation-metrics}

{Wrongful payment is payment for a nonconforming output. Unbound payment is payment without matching, authenticated, finalized, and currently valid acceptance. The replay metric records successful authentication and reproduction of the declared decision procedure. The thirteen admitted conditions comprise honest execution, nonconforming output, false approval, false rejection, evaluator silence, five malformed-receipt conditions, repeated settlement, delivery at the bound, and altered replay evidence. The five receipt conditions vary task, output, predicate, authentication, or validity while retaining the provider's conforming original submission. Eleven conditions contain conforming work and enter the receipt and payment progress denominators. Missing finality and permanent partition are the two stress conditions outside this admitted set. Delay sweeps above ten likewise test the progress boundary. Failure counts require both an observed event and eligibility under the relevant premises.}

\subsection{Repair Outcomes and Remaining Evidence}
\label{app:additional-outcomes}

Objective fallback creates replacement acceptance after evaluator silence or a contestable rejection while retaining the original record. In both architectures, it reduces source and destination receipt misses from 7/11 to 5/11. The five remaining misses arise from substitution of the task, output, or acceptance predicate, invalid authentication, and expired acceptance. Conditional P12 progress succeeds for every eligible acceptance. Recovery after silence or rejection increases the number of conforming cases with usable objective acceptance from four to six. Combined repair misses payment deadlines in 5/11 objective and 7/11 subjective conforming cases per configuration. E enables replay in twelve of thirteen admitted conditions per configuration. The altered bundle is detected and rejected, while a guaranteed authenticated recovery path from the retained data remains unestablished. The remaining unestablished requirements concern net utility (P5), complete attribution (P15), evidence recovery (P16), capability (P2), coalition gain (P7), subjective conformance (P9a), and bounded acceptance strategies (P9b).

\subsection{Sensitivity of Closure}
\label{app:closure-sensitivity}

Removing verifier soundness from $\Gamma$ leaves P9a unestablished in the four objective configurations. Removing finality or the delivery bound leaves P12 unestablished in the four cross-domain configurations. P9b remains unestablished for objective repaired configurations until an acceptance strategy is justified over the admitted histories with task-compatible response and resource bounds.

\begin{table*}[!]
\centering
\caption{Evidence and preservation checks for the reference task.}
\label{tab:closure-ledger}
\scriptsize
\setlength{\tabcolsep}{4pt}
\begin{tabular}{clll}
\toprule
\textbf{Obligation} & \multicolumn{1}{c}{\textbf{Evidence / contract}} & \multicolumn{1}{c}{\textbf{Preservation check}} & \multicolumn{1}{c}{\textbf{Failure  boundary}} \\
\midrule
P5 & Price 10; $c_{\rm P}\leq10\leq v_{\rm R}$; outside options 0 & S5: transfers, cost and delivery assumptions & Extra fees/nonpayment defeat utility bounds \\
P6 & Escrow prefunded with 10 & S5: capped payout; no double settlement & Subsidies and extra rewards need accounting \\
P9a & Sound check of the committed sorting relation & S5: authenticity, validity, all task bindings & No verifier: hole; valid receipt ignored: break \\
P9b & Justified acceptance path; callable verifier & Access, cost, schedule: acceptance by 12 & Silence or terminal refund can block access \\
P13 & Bounded delivery/scheduling; valid receipt & Payment by 16; S6: matching terminal record & Eventual delivery gives no finite deadline \\
P15 & Authenticated identities and consequential acts & S6: signatures and task-specific attribution & Signatures do not establish verdict correctness \\
P16 & Committed I/O, versions, receipt, transfers & S6: provenance and deterministic replay & Replay does not establish subjective quality \\
\bottomrule
\end{tabular}
\end{table*}

The same scope rule applies to the baseline and combined policies. The baseline preserves nine assessments, leaves 38 unestablished, and has eighteen holes and three breaks. In its two single-domain verification-first configurations, P13 fails because seven eligible cases miss the deadline.

\subsection{Delivery Boundaries, Finite Inputs, and Stress Conditions}
\label{app:finite-stress}

\begingroup

\label{app:delivery-boundaries}
The delivery sweep uses
$d\in\{0,2,5,9,10,11,12,13,14,15,16,17,18\}.$
For example, $d=14$ delivers a valid receipt at time 16 but attempts payment at time 17 after its validity has ended. Delays above 10 test behavior outside the admitted delivery bound. The sweep contains 364 objective executions across two architectures, two domain placements, and seven policies.

The finite-input check enumerates all 364 arrays of length zero through five over $\{-1,0,1\}$. Crossing two architectures, two domain placements, the baseline and combined repair, and four conditions covering honest execution, nonconforming output, silence, and false approval yields 11,648 objective executions.
\endgroup

Across the enumerated sorting inputs, each fixed architecture, domain placement, policy, and perturbation retains the same outcome pattern as the input changes. The combined repair pays conforming outputs and refuses nonconforming outputs in these checks. Within the enumerated domain, the observed failures therefore arise from the acceptance and payment interfaces rather than from a particular sorting input.

Under permanent partition, all eight configurations with the combined repair obtain timely source acceptance but miss destination delivery and payment. Missing finality instead prevents timely source acceptance. Both conditions violate premises of the affected progress guarantees, so their raw outcomes are reported separately from eligible progress cases. Every configuration eventually refunds in both conditions.

\begingroup

\section{Native Protocol Workflows}
\label{app:experiments}
\label{app:instrument}
\label{app:native-escrow}
\label{app:native-experiments}

\subsection{Registry Acceptance and Application-Level Task Binding}
\label{app:native-erc8004}

We evaluate the three unmodified ERC-8004 reference registries on a local EVM~\cite{erc8004implementation}. Thirty-six test cases cover identity authorization, reputation, validation, and task-payment linkage. Eight linkage cases bind outputs from a fixed metadata record to validator responses and vary the chain, task, output, recipient, or payment link. Six include local test-currency transfers, while two use an unavailable or missing payment reference. One case is consistent and seven alter or omit at least one required link. The registry accepts all eight submissions.

\subsection{x402 Payment Order and Service Delivery}
\label{app:http-effects}

We evaluate the official x402 implementation with its payment facilitator and an unmodified payment-token contract on a local Ethereum chain. The service computes a document digest and records the completed request. Default payment, upfront payment, and upfront payment with a payer-signed request commitment start from identical chain and application states.

{Eleven conditions cover completion, signature and recipient checks, amount and expiry checks, payment reuse, two settlement races, resource substitution, service failure, and retry. Default payment injects each race after its separate authorization check. Both upfront policies inject it immediately before the transfer attempt, which precedes service execution. We observe transfer, application effect, and response delivery separately.}

\subsection{ERC-8183 Acceptance Correctness and Availability}
\label{app:native-binding}

{We evaluate the retained ERC-8183 reference escrow on a local EVM~\cite{crapis2026erc8183}. The evaluator extensions follow its submission semantics, including attribution of the submitting provider.} Each task escrows $10^6$ base units of a local test token, with zero platform and evaluator fees and an inactive deployment administrator. The same escrow is used with address-controlled, operator-only objective, and provider-callable objective evaluation.

Tasks compute SHA-256 hashes of complete ERC-20, ERC-8004, ERC-721, or ERC-8183 documents spanning 5,601 to 41,976 bytes. An independent observer computes the reference digests. Both objective configurations commit inputs before funding and authenticate submitted digests. The evaluator checks the commitments and recomputes the digest on the EVM before accepting or rejecting the result.

\subsection{ERC-8183 Deadlines and Verification Cost}
\label{app:native-costs}

Each document has seven pre-expiry cases covering ordinary completion, approval of an incorrect digest, rejection of a correct digest, silence, repeated completion, mismatched input, and replaced submission. Three additional cases test late evaluation and both transaction orders between completion and refund. Native escrow permits completion after expiry while funds remain until refund becomes terminal. Both objective evaluator extensions instead require evaluation before expiry. None of the twelve late cases per extension receives payment.

{Under Prague EVM rules, native completion uses 81,523 gas and objective invocation uses 247,870 to 1,702,640 gas. The latter covers document transmission and hashing. A bounded-cost P9b contract must additionally compare total provider expenditure with the task budget and justify inclusion before expiry.}

\subsection{Agent Execution, Task Binding, and Settlement}
\label{app:agent-execution}

The agent experiment uses local Qwen3.5 4B and Gemma 4 E4B together with GPT-5.6 Sol, MiniMax M2.5, GLM-5.3, Gemini 3.7 Flash, Kimi K2.6, and Qwen3.6 Plus. Each model runs eight tasks through local tools in OpenCode 1.17.15~\cite{opencode2026}, giving 64 agent attempts. Four tasks hash one protocol document, while four hash two documents concatenated in the specified order without a separator. Prompts provide the documents, permitted tool commands, and required final-answer structure. Each task allows one attempt, six agent steps, 180 seconds after initialization, and at most 2,048 tokens per model response. An unparseable final answer leaves the task pending until expiry refund. Incorrect answers and timeouts retain their observed outcomes.

The same local EVM hosts the official ERC-8004 registries and ERC-8183 escrow. Each captured answer is replayed under four conditions. These preserve the original evidence and output, alter the input evidence while preserving the submitted digest, replace the output digest, or withhold the operator's verdict. Each condition is tested with address-controlled and provider-callable acceptance. Separate escrow instances hold the captured answer fixed while varying the task evidence or acceptance path.

The digest-normalization rule is fixed after the local strict results and before supplementary native execution and hosted inference. It extracts one structured submission containing the required 64-digit hexadecimal digest. Surrounding prose and repeated mentions of the same digest are allowed. Multiple objects, extra fields, conflicting digests, incomplete answers, and timeouts remain unavailable. Normalization preserves the digest value, which is then compared with the independently computed digest.

\begin{table}[!]
\centering
\begin{threeparttable}
\caption{Agent submission and settlement outcomes.}
\label{tab:agent-results}
\footnotesize
\setlength{\tabcolsep}{3pt}
\renewcommand{\arraystretch}{1.10}

\begin{tabularx}{\columnwidth}{@{}>{\hsize=1.6\hsize\linewidth=\hsize\raggedleft\arraybackslash}X *{3}{>{\hsize=.8\hsize\linewidth=\hsize\centering\arraybackslash}X}@{}}
\toprule
\multicolumn{1}{c}{\multirow{2}{*}{\textbf{Model}}} & \multirow{2}{*}{\shortstack{Correct\\tool result}} & \multicolumn{2}{c}{Correct usable submission} \\
\cmidrule(lr){3-4}
 & & S & N \\
\midrule
Qwen3.5 4B & 8/8 & 0/8 & 8/8 \\
Gemma 4 E4B & 8/8 & 8/8 & 8/8 \\
GPT-5.6 Sol & 8/8 & 8/8 & 8/8 \\
MiniMax M2.5 & 8/8 & 8/8 & 8/8 \\
GLM-5.3 & 8/8 & 8/8 & 8/8 \\
Gemini 3.7 Flash & 8/8 & 8/8 & 8/8 \\
Kimi K2.6 & 8/8 & 8/8 & 8/8 \\
Qwen3.6 Plus & 8/8 & 7/8 & 8/8 \\

\end{tabularx}

\begin{tabularx}{\columnwidth}{@{}>{\hsize=2\hsize\linewidth=\hsize\raggedleft\arraybackslash}X *{4}{>{\hsize=.75\hsize\linewidth=\hsize\centering\arraybackslash}X}@{}}
\toprule
\multicolumn{1}{c}{\multirow{2}{*}{\textbf{Condition}}} & \multicolumn{2}{c}{Address-controlled} & \multicolumn{2}{c}{Provider-callable} \\
\cmidrule(lr){2-3}
\cmidrule(lr){4-5}
 & S & N & S & N \\
\midrule
Normal & 55/64 & 64/64 & 55/64 & 64/64 \\
Input evidence replaced & 55/64 & 64/64 & 0/64 & 0/64 \\
Output digest replaced & 55/64 & 64/64 & 0/64 & 0/64 \\
Operator silence & 0/64 & 0/64 & 55/64 & 64/64 \\
\bottomrule
\end{tabularx}

\begin{tablenotes}[flushleft]
\footnotesize
\item[] S/N: strict/normalized parsing of the same answers. Correct usable submissions match the reference digest. Counts: submissions out of 8 tasks per model; payments out of 64 workflows per condition and policy.
\end{tablenotes}
\end{threeparttable}
\end{table}

\begingroup

\section{Deployment Evidence and Evaluator Decisions}
\label{app:deployment-decisions}

\subsection{ERC-8004 Feedback Integrity and Task Binding}
\label{app:public_feedback_data}

Content is retrieved for 223 referenced feedback entries and parsed for 219. Among the retrieved contents, 198 match their commitments. Of these, 157 match the retrieved bytes and 41 match after declared gzip decompression. Eight of the remaining 25 match SHA-256 rather than the specified KECCAK-256. Two structured documents reference successful Base USDC payments. Both concern one identity whose registration pointer matches its content, whose Base x402 wallet declared in the registration metadata matches the payment recipient, and whose feedback issuer matches the payer.

{The task-link search checks top-level task and job identifiers, identifiers nested under task or job descriptions, and the task identifier carried by A2A metadata. All five locations are empty in the 219 parsed documents. The two payment-linked documents also lack output commitments. Payment linkage identifies the payer from the token transfer, allowing the transaction sender to be a facilitator.}

\subsection{Evaluator Agreement and Decision Availability}
\label{app:evaluator-decisions}

The HelpSteer2 validation set contains 1,038 responses to 519 prompts. For each prompt, we retain the record with the lowest SHA-256 hash of its canonical JSON, including annotations. We then select the 24 lowest hashes from each helpfulness stratum, giving 120 distinct prompts independently of model judgments. Evaluators see the request, response, and rubric, but not the human score. Six disjoint preliminary tasks fix the inference settings.

The rubric covers relevance, correctness, completeness, and satisfaction of the user’s requirements. Scores from zero to four range from failure to address the request to full satisfaction. We treat human scores of at least three as reference-eligible and apply the same cutoff to evaluator scores.

The eight local evaluators use greedy decoding at temperature zero with seed 20260906, a 16,384-token context, and a 1,024-token generation budget. Native templates disable optional reasoning where supported. Strict parsing accepts one structured integer score from zero to four, optionally enclosed in a formatting block. Extra explanation, invalid scores, truncation, and failed calls make the decision unavailable. Among 960 local first responses, 825 are usable under strict parsing and 890 after normalization.

The hosted collection adds GPT-5.6 Sol, MiniMax M2.5, GLM-5.3, Gemini 3.7 Flash, Kimi K2.6, and Qwen3.6 Plus, giving fourteen evaluators in total. Hosted models receive the same sample and rubric through a commercial inference gateway with a requested 1,024-token generation budget. Requests do not specify a seed. GPT-5.6 Sol and Kimi K2.6 use service-default temperatures; the other hosted models use temperature zero. Optional reasoning is disabled for GPT-5.6 Sol, Kimi K2.6, and Qwen3.6 Plus, while MiniMax M2.5, GLM-5.3, and Gemini 3.7 Flash use their default reasoning modes. The strict and normalized interfaces parse the same 1,680 first responses: 960 local and 720 hosted.

We repeat the first twelve requests for each model, giving 96 local and 72 hosted repeat calls. Agreement under strict parsing requires the same integer score or unavailability in both calls. All 96 local pairs agree, including 83 pairs with the same usable score and thirteen with two unavailable outcomes. Hosted agreement ranges from 7/12 to 12/12. Primary comparisons use the original first responses.

\subsection{Committee Composition and Voting Rules}
\label{app:paired-diagnostics}

Unanimity remains pending if any vote is unavailable. With all votes usable, it approves only if every member approves and otherwise rejects. Strict majority requires more than half of the full membership to vote in one direction and otherwise remains pending. We apply both voting rules to every set of one, two, three, five, or seven members from the eight local models under both interfaces. This gives 624 configurations evaluated on the same 120 requests. For normalized strict-majority committees of odd sizes, below-reference approvals range from $[19,67]$ for individual models to $[42,50]$ for seven-member committees, with medians between 44 and 45. Larger committees narrow the observed range while retaining a similar median. Because the committees share models and requests, these distributions measure sensitivity to membership within the observed local-model collection.

\label{app:evaluator-settlement}

\section{Economic Models, Statements, and Proofs}
\label{app:econ-full}
\label{app:proofs}
\begingroup

\subsection{Persuadable Report Model}
\label{sec:econ-model}
\label{app:report-context}

We model a \DAE interaction as a sequential game governed by a mechanism $M$, with $n$ participants and an adversary $\mathcal{A}$. Participant $i$ has private type $v_i$, such as a valuation or cost, and submits $b_i=\beta_i(v_i,m)$ after receiving context $m$. Classical mechanism design treats this report policy as part of the participant's strategy. Under A1, counterparty content can instead change the report produced by a fixed policy. We consider the persuadable report class
\begin{equation}
\label{eq:persuadable}
\beta_i(v_i,m)=v_i+\delta_i(m), \qquad
\delta_i(m)\in[-\delta,\delta],
\end{equation}
where $\delta\geq0$ bounds the change in the realized report and $m$ is adversary-controlled message or context content.

The model holds task-relevant facts fixed. Let $\phi(m)$ retain those facts, and require admissible messages to satisfy $\phi(m)=\phi(m_0)$ for a fact-only reference $m_0$. For Proposition~\ref{prop:ic}, the reference report is truthful, so $\delta_i(m_0)=0$. For a general policy, the message-induced change is instead measured against that policy's reference report,
$\Delta_i(v_i,m)=\beta_i(v_i,m)-\beta_i(v_i,m_0).$
The difference $\beta_i(v_i,m_0)-v_i$ represents strategic reporting already present before persuasion.

\subsection{Generative Extractable Value}
\label{sec:econ-gev}
\label{app:econ-notes}

\begin{definition}[\GEV]
\label{def:gev}
Fix an interaction with mechanism $M$ and a reference profile for all participants other than a distinguished party, and let $V(\sigma)$ denote the expected gross economic value captured by that party under strategy $\sigma$, before the cost of acquiring adversarial control. Let $\Sigma_{\mathrm{ord}}$, $\Sigma_{\mathrm{msg}}$, and $\Sigma_{\mathrm{ctx}}$ denote nested strategy classes that additionally permit control over transaction ordering, message content visible to counterparty \LLMs, and counterparty retrieval or context assembly. Assume the value supremum over each of the four strategy classes is finite. \emph{Generative Extractable Value} is
\begin{equation*}
\mathrm{GEV}
=
\sup_{\sigma\in\Sigma_{\mathrm{ctx}}}V(\sigma)
-
\sup_{\sigma\in\Sigma_{\varnothing}}V(\sigma),
\end{equation*}
where $\Sigma_{\varnothing}\subset\Sigma_{\mathrm{ord}}\subset\Sigma_{\mathrm{msg}}\subset\Sigma_{\mathrm{ctx}}$, and $\Sigma_{\varnothing}$ is the honest-participation baseline in which the distinguished party follows its protocol role, messages contain only the task-relevant fact projection, and the retrieval corpus is fixed. The increment from $\Sigma_{\varnothing}$ to $\Sigma_{\mathrm{ord}}$ corresponds to the additional value obtained through transaction ordering, as studied in MEV \cite{daian2020flash,babel2023clockwork}.
\end{definition}

\subsection{Incentive Compatibility Under Persuadable Bidders}
\label{sec:econ-ic}
\label{app:auction-context}

\begin{proposition}[Report distortion under persuasion]
\label{prop:ic}
Let $M$ be a sealed-bid second-price auction with at least two bidders, no reserve or report cap, and no withdrawal after reports are submitted. Bidder $i$ has private value $v_i$ and submits bid $b_i$. Let bidders lie in the persuadable class \eqref{eq:persuadable} with $\delta>0$. Assume the relevant type and report neighborhoods are admissible and the adversary can realize any shift profile in $[-\delta,\delta]^n$. Then
(a) No bidder can guarantee $b_i=v_i$ against this adversary class. In a realization, $b_i\neq v_i$ exactly when $\delta_i(m)\neq0$, and an observed bid implies only that the bidder's type lies within radius $\delta$.
(b) Against an adversary that knows every bidder's valuation, the allocation can incur worst-case welfare loss approaching $2\delta$, and no strategy in the class can lose more.
(c) A seller-affiliated adversary with bidder-specific message access can extract $\delta$ per interaction without observing the value profile, because shifting every report by $\delta$ raises the clearing price by exactly $\delta$ while preserving the allocation. The witnessing strategy uses no ordering action, so $\mathrm{GEV}\geq\delta$ while the auction itself executes correctly. Under zero ordering premium, the same bound applies to the marginal message increment.

Moreover, when the type space is closed under shifts of size $\delta$, a mechanism observing only submitted reports cannot distinguish a persuaded profile from truthful play by a correspondingly shifted type profile.
\end{proposition}

\keypoint{Relation to auction results}
Proposition~\ref{prop:ic}(b) derives a welfare bound directly from bounded report errors. Chiesa, Micali, and Zhu analyze a related interval-uncertainty problem with uncertain own valuations and undominated strategies \cite{chiesa2015knightian}. Our model fixes true values and allows outside messages to distort reports. Bid-taker cheating and credible-auction work study other manipulation channels \cite{rothkopf1995,rothkopf1990rare,porter2003cheating,akbarpour2020credible}. Our common-shift construction leaves the auction rule unchanged and requires the specified message-induced shift to be attainable.

\begin{IEEEproof}[Proof of Proposition~\ref{prop:ic}]
Let $M$ be the sealed-bid second-price auction. The highest report wins and pays the second-highest report, with ties resolved by a fixed rule.

\emph{(a) Truthful reporting under adversarial perturbation.}
The realized report of bidder $i$ is
$b_i=v_i+\delta_i(m).$
Whenever the adversary induces a nonzero shift, the submitted report differs from $v_i$. The bidder cannot guarantee a truthful report against a class that permits such a shift. Rearranging the report equation gives
$v_i\in[b_i-\delta,b_i+\delta].$
Under the vector-targeted assumption in Prop.~\ref{prop:ic}, every value in this interval is compatible with the report when the whole interval lies in the type space. The same report vector $b$ also arises from truthful reporting by types $v'=b$. A report-only mechanism cannot distinguish the two cases when the relevant type neighborhoods are admissible.

\emph{(b) Welfare loss.}
Let $n=2$ with
$v_1=v, v_2=v-\epsilon,$
where $0<\epsilon<2\delta$. The efficient outcome allocates the item to bidder 1. Choose shifts
$(\delta_1,\delta_2)=(-\delta,+\delta')$
with $\delta'\leq\delta$ and $\delta'>\epsilon-\delta$. Then
$b_2=v-\epsilon+\delta' > v-\delta=b_1,$
so bidder 2 wins and welfare falls by $\epsilon$. Letting $\epsilon\uparrow2\delta$ gives a loss approaching $2\delta$.

For tightness, let $i^*$ be the efficient bidder and $w$ the realized winner. Because $b_w\geq b_{i^*}$ and every report differs from its value by at most $\delta$,
$v_w \geq b_w-\delta \geq b_{i^*}-\delta \geq v_{i^*}-2\delta.$
Thus no realization in the class loses more than $2\delta$.

\emph{(c) Message-channel extraction.}
Take the seller as the distinguished party, $V$ as auction revenue, and truthful reports as the honest baseline. A seller-affiliated adversary can apply the common shift
$d=(+\delta,\dots,+\delta)$
without observing the value profile. Adding the same constant to every report preserves the report ordering and therefore preserves the allocation. The second-highest report increases by exactly $\delta$, so the clearing price increases by $\delta$. The seller gains $\delta$ while the witnessing strategy uses no ordering action. By Def.~\ref{def:gev},
$\mathrm{GEV}\geq\delta.$

\end{IEEEproof}

\begin{corollary}[Restoration by content-minimized bidding]
\label{cor:restore}
Suppose verifiable bid derivation has the form $b_i=f_i(v_i,\tau)$ for a deterministic policy $f_i$ committed before interaction. Require every pair of fact-equivalent messages to yield the same canonical, attested context $\tau$, and require $v_i$ to be fixed before and independently of the adversary's message. Then $\Delta_i=0$, establishing P8 while allowing baseline strategic shading. In the private-value second-price auction $M$ of Prop.~\ref{prop:ic}, where utility equals the value of the allocated item minus payment, for each fixed profile of rival reports, truthful $f_i(v,\tau)=v$ yields utility at least as high as any other admissible report policy. When all bidders choose this policy, bids and outcomes coincide with truthful execution of $M$.
\end{corollary}

\begin{IEEEproof}[Proof of Corollary~\ref{cor:restore}]
Fix the committed deterministic policy $f_i$ and private value $v_i$. Fact-equivalent messages produce the same canonical attested context $\tau$, so $f_i(v_i,\tau)$ is unchanged. This proves $\Delta_i=0$ and P8 for every committed policy in the stated class, without requiring that policy to report truthfully.

Under the stated private-value, quasilinear second-price model, reporting $v_i$ weakly dominates any alternative for every fixed profile of rival reports \cite{myerson1981,nisan2007}. The truthful policy therefore yields at least as much utility as any other admissible policy against that profile. If all bidders choose it, bids equal true values and the usual allocation and payment follow.
\end{IEEEproof}

\subsection{Collusion Cost of Bonded Evaluator Committees}
\label{sec:econ-eval}

\begin{proposition}[Evaluator corruption cost]
\label{prop:evaluator}
Consider escrow of value $p$ released on evaluator approval. A wrongful release yields provider gain $g\leq p$.

(a) \emph{Interested-party baseline.} Consider a single unbonded evaluator selected by an interested party, as permitted by the ERC-8183 core design and by an ACP configuration with one requester-appointed evaluator and no additional bond \cite{crapis2026erc8183,virtuals-acp}. If its modeled exposure is its expected future participation value $r$, compensation of at most $r$ offsets a false approval. For new identities with $r\rightarrow0$, this corruption cost does not scale with $p$. When the evaluator and client are the same party, wrongful rejection requires no external bribe.

(b) \emph{Bonded random committee.} Select $n$ evaluators using a verifiable random function (VRF) \cite{micali1999vrf} from identities that have posted bond $b$, with each selected bond committed separately for the job. Release requires $t$ approvals, where $1\leq t\leq n$. Honest members reject nonconforming work in this attack class and retain their bonds. False approval brings no additional protocol reward or cost saving. Let $q\in(0,1]$ uniformly lower-bound the probability that each false approving vote incurs an enforced net bond loss of $b$ within the job horizon, regardless of quorum formation. Net penalties leave the attacking coalition and cannot be recovered through its identities. For any attack $A$ through bribed or adversary-owned seats, with wrongful-release event $W$,
$\mathbb E[\operatorname{cost}(A)]\geq tqb\,\Pr[W\mid A].$
Thus $b>g/(tq)$ makes every such attack with positive success probability unprofitable. Adversary-owned identities must be funded before selection. Evaluators maximize undiscounted expected monetary payoff, and promised bribes are enforceable.
\end{proposition}

\begin{IEEEproof}[Proof of Proposition~\ref{prop:evaluator}]
We use expected monetary payoffs and enforceable side payments.

\emph{(b) Bonded random committee.}
For a fixed attack $A$, let $N$ count false approving votes. Each vote has expected net bond loss at least $qb$. Honest rejection retains the bond, and false approval brings no extra protocol reward or cost saving. A bribed evaluator therefore needs compensation for this loss, while an adversary-owned seat bears it directly. Adding the costs over seats gives
$\mathbb E[\operatorname{cost}(A)]\geq qb\,\mathbb E[N].$
Wrongful release needs $N\geq t$, so $N\geq t\mathbf1_W$, where $\mathbf1_W$ indicates wrongful release. Hence
$\mathbb E[\operatorname{cost}(A)]\geq tqb\,\Pr[W\mid A].$
The attacker's expected profit is at most
$(g-tqb)\Pr[W\mid A],$
which is negative when $b>g/(tq)$ and success has positive probability. This argument permits correlated votes and randomized attack initiation.
\end{IEEEproof}

\subsection{The Correlation Floor on Committee Soundness}
\label{sec:econ-corr}
\label{app:correlation-proof}

\begin{proposition}[Correlation floor]
\label{prop:correlated}
Let $n$ evaluators vote by majority on a deliverable, with ties counted as error. For a fixed conformance label and a declared distribution of tasks and attacks, each evaluator has marginal error probability $0<e<1/2$ and every pair has error correlation $\rho\in[0,1)$. To interpret this as a P9a false-approval bound, use the distribution conditional on nonconforming deliverables.

(a) The effective committee size $n_{\mathrm{eff}}$ and the variance $\sigma^2$ of the fraction of erroneous votes are
$n_{\mathrm{eff}}=\frac{n}{1+(n-1)\rho}, \qquad \sigma^2=\frac{e(1-e)}{n_{\mathrm{eff}}}.$
The majority-error probability is at most
$\frac{\sigma^2}{\sigma^2+(1/2-e)^2}$
by Cantelli's inequality. This upper bound approaches a nonzero limit whenever $\rho>0$.

(b) There exist committees with the specified marginal error and pairwise correlation whose majority-error probability is at least $\rho e$ for every $n$. Consequently, committee size alone cannot certify an error probability below $\rho e$ over this class.

\end{proposition}

\begin{IEEEproof}[Proof of Proposition~\ref{prop:correlated}]
Let $X_1,\dots,X_n\in\{0,1\}$ indicate evaluator errors. Ties count as majority error. The common marginal error and pairwise correlation satisfy
$0<\Pr[X_i=1]=e<\frac12$
and
$\mathrm{Corr}(X_i,X_j)=\rho$
for $i\neq j$. Define the total and fraction of erroneous votes by
$S=\sum_iX_i,\qquad \bar X=\frac{S}{n}.$

\emph{(a) Correlation-aware upper bound.}
Each $X_i$ has variance $e(1-e)$ and each distinct pair has covariance $\rho e(1-e)$. Therefore
\begin{align*}
\mathrm{Var}(\bar X)
&=
\frac{1}{n^2}
\left(
n e(1-e)+n(n-1)\rho e(1-e)
\right) \nonumber\\
&=
\frac{e(1-e)(1+(n-1)\rho)}{n}
=
\frac{e(1-e)}{n_{\mathrm{eff}}},
\end{align*}
where
$n_{\mathrm{eff}} = \frac{n}{1+(n-1)\rho}.$
The majority-error bound below specializes Ladha's Proposition~1, Eq.~(2), to equal error probabilities, with its Cantelli proof on p.~633 \cite{ladha1992}. Replacing correctness indicators by $1-X_i$ gives mean correctness $1-e$ and the same variance. Since the majority fails when $\bar X\geq1/2$, Cantelli's inequality yields
$\Pr[\bar X\geq1/2] \leq \frac{\sigma^2} {\sigma^2+(1/2-e)^2}, \qquad \sigma^2=\frac{e(1-e)}{n_{\mathrm{eff}}}.$

\emph{(b) Common-source lower bound.}
Draw
$Z\sim\mathrm{Bernoulli}(\rho).$
When $Z=1$, draw one
$C\sim\mathrm{Bernoulli}(e)$
and set $X_i=C$ for every evaluator. When $Z=0$, draw the $X_i$ independently as Bernoulli$(e)$. The marginal error probability remains $e$, and
$\mathrm{Corr}(X_i,X_j)=\rho.$
On event $\{Z=1,C=1\}$, which has probability $\rho e$, every evaluator errs. Hence
$\Pr[\text{majority error}] \geq \rho e$
for every $n$.

\end{IEEEproof}

\begin{corollary}[Combined corruption limit]
\label{cor:combined}
Let $W$ denote wrongful release within one job. For attack $A$ with $\Pr[W\mid A]>0$, define its expected cost per unit of success probability by
$K(A)=\frac{\mathbb{E}[\operatorname{cost}(A)]}{\Pr[W\mid A]}.$
Cost counts the attacker's undiscounted net payment obligations incurred for the job, regardless of payment date, together with resource expenditure and losses on adversary-owned bonds within the job horizon. Let $C_{\mathrm{dir}}$ and $C_{\mathrm{err}}$ be the infima of $K(A)$ over pure purchased-seat and pure induced-error strategy classes. Proposition~\ref{prop:evaluator}(b) gives $K(A)\geq tqb$ for every direct attack with positive success probability, hence $C_{\mathrm{dir}}\geq tqb$. The infimum over the two pure channels is
$C_{\mathrm{pure}} = \min\{C_{\mathrm{dir}},C_{\mathrm{err}}\} \geq \min\{tqb,C_{\mathrm{err}}\}.$
The induced-error channel also gives $C_{\mathrm{pure}}\leq C_{\mathrm{err}}$. When $C_{\mathrm{err}}<\infty$, increasing the bond cannot certify a pure-channel attack cost above $C_{\mathrm{err}}$ without also raising the cost of induced common-source failure. For attacks that combine seat purchase and induced errors, the corresponding bound depends on a joint model of both actions and the evidence supporting penalties.
\end{corollary}

\endgroup

\raggedbottom
\endgroup
\endgroup
\endgroup

\end{document}